\documentclass[leqno,a4paper,twoside,11pt]{article}
\usepackage[nointlimits,nosumlimits]{amsmath}
\usepackage{amsfonts,amssymb,amsthm,ifthen}
\usepackage{color}
\usepackage[arrow,matrix,curve]{xy}
\usepackage{hyperref}
\usepackage[utf8x]{inputenc}
\usepackage[square,numbers]{natbib}

\newtheorem{Thm}{Theorem}[section]
\newtheorem{Prp}[Thm]{Proposition}
\newtheorem{Lem}[Thm]{Lemma}
\newtheorem{Cor}[Thm]{Corollary}
\newtheorem{Def}[Thm]{Definition}

\newtheorem{IntroThm}{Theorem}[section]

\theoremstyle{definition}

\newcommand{\bC}{\mathbb C}
\newcommand{\mE}{\mathcal E}
\newcommand{\bN}{\mathbb N}
\newcommand{\bR}{\mathbb R}
\newcommand{\bZ}{\mathbb Z}
\newcommand{\mO}{\mathcal O}
\newcommand{\mT}{\mathcal T}
\newcommand{\Ber}{\mathrm{Ber}}
\newcommand{\sdet}{\mathrm{sdet}}
\newcommand{\hol}{\mathrm{hol}}
\newcommand{\Hol}{\mathrm{Hol}}
\newcommand{\id}{\mathrm{id}}
\newcommand{\rk}{\mathrm{rk}}
\newcommand{\str}{\mathrm{str}}
\newcommand{\dbar}{\overline{\partial}}
\newcommand{\oalpha}{\overline{\alpha}}
\newcommand{\ogamma}{\overline{\gamma}}
\newcommand{\oxi}{\overline{\xi}}

\newcommand{\setsep}{\;\big|\;}
\newcommand{\dd}[2][]{\frac{\partial #1}{\partial #2}}
\newcommand{\abs}[1]{\left| #1 \right|}
\newcommand{\scal}[3][]{\ifthenelse{\equal{#1}{}}{
  \left\langle #2,\,#3 \right\rangle
}{\ifthenelse{\equal{#1}{(}}{
  \left( #2,\,#3 \right)
}{\ifthenelse{\equal{#1}{[}}{
  \left[ #2,\,#3 \right]
}{
  #1\left( #2,\,#3 \right)
}}}}

\newcommand\blfootnote[1]{%
  \begingroup
  \renewcommand\thefootnote{}\footnotetext{#1}%
  \endgroup
}

\renewcommand{\title}[1]{\vbox{\center\LARGE{\textsc{#1}}}\vspace{5mm}}
\renewcommand{\author}[1]{\vbox{\center\large{\textsc{#1}}}\vspace{5mm}}
\newcommand{\address}[1]{\vbox{\center\em#1}}
\newcommand{\email}[1]{\vbox{\center\tt#1}\vspace{5mm}}

\begin{document}

\title{On Complex Supermanifolds\\ with Trivial Canonical Bundle}

\author{Josua Groeger$^1$}

\blfootnote{Research funded by the Institutional Strategy of the University
of Cologne in the German Excellence Initiative.}

\address{Universit\"at zu K\"oln, Institut f\"ur Theoretische Physik,\\
  Z\"ulpicher Str. 77, 50937 K\"oln, Germany }

\email{$^1$groegerj@thp.uni-koeln.de}

\begin{abstract}
\noindent
We give an algebraic characterisation for the triviality of the
canonical bundle of a complex supermanifold in terms of a certain
Batalin-Vilkovisky superalgebra structure.
As an application, we study the Calabi-Yau case, in which
an explicit formula in terms of the Levi-Civita connection is achieved.
Our methods include the use of complex integral forms and
the recently developed theory of superholonomy.
\end{abstract}

\section{Introduction}

The structure of a Batalin-Vilkovisky (BV) algebra was first found
in the perturbative solutions of the quantum master equation,
which is demanded by BRS invariance in the quantisation of gauge theories
\cite{BV81}.
Through a formal analogy with the Maurer-Cartan equation of complex
deformations on a complex manifold $M$, a certain BV structure also arises
as a necessary condition for the existence of infinitesimal solutions \cite{BK98},
and is tightly related to the canonical bundle of $M$ being trivial.

In this article, we are interested in the case of a complex supermanifold $M$
of dimension $\dim M=(\dim M)_{\overline{0}}|(\dim M)_{\overline{1}}=n|m$.
$M$ is assumed connected and, as usual, we let $M_{\overline{0}}$ denote
the underlying complex manifold, and $\mO_M$ the sheaf of
holomorphic superfunctions. By a slight abuse of notation, the superalgebra
of global sections shall be denoted by the same symbol.
We define the \emph{canonical bundle of $M$} to be the Berezinian
of the complex cotangent sheaf $\Ber M:=\Ber(\mT^{1,0}M)^*$, which
transforms through the superdeterminant $\sdet d\varphi$ under a holomorphic
change of coordinates $\varphi$, and, in the case of an ordinary complex manifold,
reduces to the classical canonical bundle of top degree holomorphic forms.
Moreover, we consider the following sheaf of $(0,q)$-forms with values
in holomorphic multivector fields.
\begin{align}
\label{eqnMultivectorForms}
\Omega^{0,*}\left(\bigwedge^*\mT^{1,0}M\right)
:=\bigoplus_{q=0}^{\infty}\bigoplus_{p=0}^{\infty}\Omega^{0,q}\left(\bigwedge^p\mT^{1,0}M\right)
\end{align}
Our first main theorem, to be stated next, gives an algebraic characterisation
for the triviality of the canonical bundle, that is
$\Ber M\cong\mO_M^{1|0}$ for $m$ even and $\Ber M\cong\mO_M^{0|1}$ for $m$ odd,
respectively.

\begin{IntroThm}
\label{thmWeakCYBV}
Let $M$ be a simply connected complex supermanifold.
Then the superalgebra $\Omega^{0,*}(\bigwedge^*\mT^{1,0}M)$
carries the structure of
a differential Gerstenhaber-Batalin-Vilkovisky (dGBV) superalgebra
strongly compatible with the Schouten-Nijenhuis bracket
if and only if $\Ber M$ is trivial.

In this case, there is a 1:1-correspondence between trivialising
homogeneous global
sections $\omega\in\Ber M$ (up to a complex constant)
and such dGBV-structures.
\end{IntroThm}

For the case $m=(\dim M)_{\overline{1}}=0$ of a complex manifold,
this theorem is a classical result \cite{Sch98}.
The superisation in the present form is the major achievement of this article.
It only became available thanks to the recently developed
superholonomy theory (\cite{Gal09,Gro14c,Gro16}).
We remark that Thm. \ref{thmWeakCYBV} justifies, a posteriori,
our definition of the canonical bundle of a complex supermanifold.

We comment on notions and structure of the proof, and of the article.
To begin with, we state the general definition of a dGBV superalgebra in
the first part of Sec. \ref{secBV}. In particular, we carefully
explain our conventions used.
In the second part of that section, we treat the case of dGBV structures
on the superalgebra (\ref{eqnMultivectorForms}). In particular, we analyse
compatibility conditions with the Schouten-Nijenhuis bracket, which
is a natural extension of the vector field bracket.

One direction of the proof of Thm. \ref{thmWeakCYBV} is established as follows.
Starting with a trivialisation of the canonical bundle,
the dGBV superalgebra structure is obtained
through a super extension of an identity stated in \cite{BK98} which,
in turn, generalises a lemma proved independently by both
Tian \cite{Tia87} and Todorov \cite{Tod89}.
This is the subject matter of Sec. \ref{secTianTodorov} with the main
result contained in Prp. \ref{prpBV}.
For this implication, it is not necessary to assume that $M$ be simply connected.

Conversely, we construct and study a connection associated with a given dGBV
structure in Sec. \ref{secFlatConnections}, which turns out to be flat.
The remaining direction of Thm. \ref{thmWeakCYBV} is then obtained by means of
the aforementioned superholonomy theory.
More precisely, the Holonomy Principle applied to the present situation
provides us with a global section $\omega\in\Ber M$ as advertised,
parallel with respect to the aforementioned connection
(Prp. \ref{prpBVParallelSection}).

In Sec. \ref{secCY}, we sudy semi-Riemannian supermanifolds with
holonomy contained in some special unitary supergroup.
This condition will be referred to as 'Calabi-Yau', even though
the Calabi-Yau theorem does not generalise to supergeometry
(we refer to the remarks in that section for details).
The Levi-Civita connection on a Calabi-Yau supermanifold induces
a connection on $\Ber M$ with trivial holonomy and, therefore,
a parallel trivialising section. In this case, all constructions
are natural, as established in our second main theorem as follows.
A more precise version is provided as Thm. \ref{thmCYConnections} below.

\begin{IntroThm}
\label{thmIntroCYConnections}
The flat connection on $\Ber M$ constructed in Sec. \ref{secFlatConnections}
with respect to the dGBV structure of Sec. \ref{secTianTodorov}
on a Calabi-Yau supermanifold coincides with the one
induced by the Levi-Civita connection.
In particular, the dGBV structure can be explicitly expressed in terms
of the Levi-Civita connection.
\end{IntroThm}

Throughout this article, we have tried not to disrupt the main flow of argument
longer than necessary. Mainly, but not only, for this reason
we have compiled an independent appendix, App. \ref{secComplexSupergeometry},
on selected elements of complex supergeometry, with
definitions, conventions and results contained for easy reference in the main text.
Topics include a synopsis on supermanifolds and vector bundles,
definition and properties of the Schouten-Nijenhuis bracket,
parallel transport and superholonomy and, finally, an exposition on
the canonical bundle and integral forms.

\section{Batalin-Vilkovisky Superstructures}
\label{secBV}

This section introduces the general notions of BV and dGBV superalgebras
first, followed by a brief account on algebraic properties
in the cases based on our main example (\ref{eqnMultivectorForms})
of a superalgebra.

The structures presented are meant to be superisations
of their classical counterparts.
We remark that several variants of the latter exist in the literature,
of which we mainly follow \cite{Man99}, \cite{Huy05}.
If not stated otherwise, a superalgebra will be either real or complex,
and the notion of linearity refers to the choice of scalars.

On the superalgebra (\ref{eqnMultivectorForms})
which is supposed to carry a BV structure, one observes that there
are two sorts of $\bZ_2$-gradings involved: The cohomological degree
$q+p$ as well as the parity ('super degree') of objects induced
by the parities of vectors and covectors.
In general, such a situation can be modelled by a \emph{graded superalgebra},
which is a vector space with a bigrading and a correponding multiplication.
Another possibility is to combine the two degrees to a single one
and consider ordinary super (i.e. $\bZ_2$-graded) algebras.
It has been argued in \cite{DM99}
that both approaches are equivalent, although leading to different signs.
They are referred to as
''point of view I'' and ''point of view II'' in that reference.

Throughout this article, we shall consistently adopt the first point of view.
This affects, in particular, the rest of this section but also parts of
the appendix, App. \ref{secComplexSupergeometry}.
We use the following convention: The 'super degree'
of an object $X$ is denoted $\abs{X}\in\bZ_2$, while its cohomological
degree is referred to as $\deg{X}\in\bZ_2$. To give an example,
we state the supercommutativity rule for a graded superalgebra with this notation.
\begin{align*}
a\cdot b=(-1)^{\deg(a)\deg(b)+\abs{a}\abs{b}}b\cdot a
\end{align*}

\begin{Def}
We call a linear map $f$ between graded superalgebras \emph{deg-odd},
if it is odd (i.e. parity-reversing) with respect to the cohomological degree
and even (i.e. parity-preserving) with respect to the super degree.
\end{Def}

We proceed with our main definitions.

\begin{Def}
A \emph{Batalin-Vilkovisky (BV) superalgebra} is a pair $(A,\Delta)$,
where $A$ is a supercommutative graded complex superalgebra,
and $\Delta:A\rightarrow A$ is a deg-odd complex linear map
such that, for every $\alpha\in A$, the map
\begin{align*}
\delta_{\alpha}:A\rightarrow A\;,\qquad
\beta\mapsto(-1)^{\deg(\alpha)}\Delta(\alpha\cdot\beta)-(-1)^{\deg(\alpha)}\Delta(\alpha)\cdot\beta-\alpha\cdot\Delta(\beta)
\end{align*}
is a derivation of degree $\deg(\alpha)+1$ and super degree the same as $\alpha$.
In other words, it satisfies
\begin{align*}
\delta_{\alpha}(\beta\cdot\gamma)=\delta_{\alpha}(\beta)\cdot\gamma+(-1)^{(\deg(\alpha)+1)\deg(\beta)+\abs{\alpha}\abs{\beta}}\beta\cdot\delta_{\alpha}(\gamma)
\end{align*}
\end{Def}

\begin{Def}
We say that a BV superalgebra $(A,\Delta)$ is \emph{compatible} with
a bracket, i.e. with a complex bilinar map
$\scal[[]{\cdot}{\cdot}:A\times A\rightarrow A$
if $\scal[[]{\alpha}{\beta}=-\delta_{\alpha}(\beta)$ for all $\alpha,\beta\in A$.
\end{Def}

We remark that the sign in the previous definition may be replaced by any
constant $c\in\bC^*$ upon redefining $\Delta$ accordingly.
The present choice is consistent with our conventions for the operators
occurring in the generalised
Tian-Todorov lemma, Lem. \ref{lemTianTodorov}, below.

\begin{Def}
A BV superalgebra $(A,\Delta)$ is called \emph{Gerstenhaber-Batalin-Vilkovisky (GBV)}
if $\Delta\circ\Delta=0$.
\end{Def}

The following statement is a straightforward consequence of the axioms.

\begin{Lem}
\label{lemGBVCompatible}
On a GBV superalgebra $(A,\Delta)$ which is compatible with a bracket
$\scal[[]{\cdot}{\cdot}$, the following equation is satisfied.
\begin{align*}
-\Delta(\scal[[]{\alpha}{\beta})=-\scal[[]{\Delta(\alpha)}{\beta}+(-1)^{\deg(\alpha)}\scal[[]{\alpha}{\Delta(\beta)}
\end{align*}
\end{Lem}

\begin{proof}
Using $\Delta^2=0$, the compatibility equation
$\scal[[]{\alpha}{\beta}=-\delta_{\alpha}(\beta)$ with $\alpha$ replaced
by $\Delta(\alpha)$, and with $\beta$ replaced by $\Delta(\beta)$, respectively,
reads as follows.
\begin{align*}
-(-1)^{\deg(\alpha)}\Delta(\Delta(\alpha)\cdot\beta)&=-\scal[[]{\Delta(\alpha)}{\beta}+\Delta(\alpha)\cdot\Delta(\beta)\\
-\Delta(\alpha\cdot\Delta(\beta))&=(-1)^{\deg(\alpha)}\scal[[]{\alpha}{\Delta(\beta)}-\Delta(\alpha)\cdot\Delta(\beta)
\end{align*}
Similarly, applying $\Delta$ on both sides of the compatibility equation,
we obtain
\begin{align*}
-\Delta(\scal[[]{\alpha}{\beta})=-(-1)^{\deg(\alpha)}\Delta(\Delta(\alpha)\cdot\beta)-\Delta(\alpha\cdot\Delta(\beta))
\end{align*}
With the previous formulas, the statement follows.
\end{proof}

\begin{Def}
A \emph{differential Gerstenhaber-Batalin-Vilkovisky (dGBV) superalgebra}
is a triple $(A,\Delta,d)$ such that $(A,\Delta)$ is a GBV algebra
and $d$ is a deg-odd complex linear derivation of degree $1$, i.e.
such that
\begin{align*}
d(\alpha\cdot\beta)=d(\alpha)\cdot\beta+(-1)^{\deg(\alpha)}\alpha\cdot d(\beta)
\end{align*}
and, in addition, such that $d\circ d=0$ and $d\circ\Delta+\Delta\circ d=0$.
\end{Def}

\subsection{The Schouten-Nijenhuis Case}

Having introduced the general definitions, we now consider dGBV superalgebras
of the form
$\left(\Omega^{0,*}(\bigwedge^*\mT^{1,0}M),\Delta,\dbar\right)$
which are compatible with the Schouten-Nijenhuis bracket as defined in
Def. \ref{defSchoutenBracket}.
The latter property we shall simply refer to as \emph{compatible} in the following.
The only datum to be specified in the case at hand is an appropriate operator
$\Delta:\Omega^{0,*}(\bigwedge^*\mT^{1,0}M)\circlearrowleft$.
By the following simple observation, the definition of such an operator
on the entire superalgebra is highly redundant.

\begin{Lem}
\label{lemSchoutenBVGenerators}
Let $\left(\Omega^{0,*}(\bigwedge^*\mT^{1,0}M),\Delta,\dbar\right)$
be a dGBV superalgebra compatible with the Schouten-Nijenhuis bracket.
Then $\Delta$ is uniquely determined by its values on elements
of the following types, for every open subset $U\subseteq M_{\overline{0}}$:
Functions $f\in\mO_M(U)$, tangent vectors
$X\in\mT^{1,0}M(U)$, and forms $\lambda\in\Omega^{0,1}M(U)$.
\end{Lem}

\begin{proof}
Locally, that is upon restriction to a sufficiently small open subset
$U\subseteq M_{\overline{0}}$,
every element of $\Omega^{0,*}(\bigwedge^*\mT^{1,0}M)$ can be written as a
polynomial in functions, tangent vectors and $(0,1)$-forms.
By linearity, it suffices to consider monomials.
Having specified $\Delta$ restricted to the according subspaces,
the original operator is recursively determined by the
Schouten-Nijenhuis bracket compatibility condition in the following form.
\begin{align*}
\Delta(\alpha\cdot\beta)=-(-1)^{\deg(\alpha)}\scal[[]{\alpha}{\beta}+\Delta(\alpha)\cdot\beta+(-1)^{\deg(\alpha)}\alpha\cdot\Delta(\beta)
\end{align*}
\end{proof}

The following condition is the one arising in Thm. \ref{thmWeakCYBV}.

\begin{Def}
A dGBV superalgebra $\left(\Omega^{0,*}(\bigwedge^*\mT^{1,0}M),\Delta,\dbar\right)$
is called \emph{strongly compatible} (with the Schouten-Nijenhuis bracket) if
it is compatible and $\Delta$ is a direct sum of operators
\begin{align*}
\Delta^{p,q}:\Omega^{0,q}\left(\bigwedge^p\mT^{1,0}M\right)\rightarrow\Omega^{0,q}\left(\bigwedge^{p-1}\mT^{1,0}M\right)
\end{align*}
\end{Def}

In other words, we demand $\Delta$ to treat the $\bZ$-degree (that we otherwise
do not consider) analogous to the Schouten-Nijenhuis bracket.
The characterisation in terms of Lem. \ref{lemSchoutenBVGenerators} reads as follows.

\begin{Cor}
A compatible dGBV superalgebra is strongly compatible if and only if
$\Delta$ has the following properties.
$\Delta(f)=0$ for $f\in\mO_M(U)$ and
$\Delta(X)\in\mO_M(U)$ for $X\in\mT^{1,0}M(U)$ and $\Delta(\lambda)=0$ for
$\lambda\in\Omega^{0,1}M(U)$,
for every open subset $U\subseteq M_{\overline{0}}$.
\end{Cor}

As to be elaborated now, the operator $\Delta$ of a compatible dGBV superalgebra
can be projected such as to give rise to a strongly compatible dGBV superalgebra.
Let $\Delta:\Omega^{0,*}(\bigwedge^*\mT^{1,0}M)\circlearrowleft$
be a complex linear operator. We let
$\tilde{\Delta}:\Omega^{0,*}(\bigwedge^*\mT^{1,0}M)\circlearrowleft$
denote the direct sum of projections
\begin{align*}
\tilde{\Delta}^{p,q}:=\Pi_{p-1,q}\circ\Delta|_{\Omega^{0,q}(\bigwedge^p\mT^{1,0}M)}\;,\qquad
\Pi_{p-1,q}:=\Pi_{\Omega^{0,q}(\bigwedge^{p-1}\mT^{1,0}M)}
\end{align*}
corresponding to the decomposition (\ref{eqnMultivectorForms}).

\begin{Lem}
Let $\left(\Omega^{0,*}(\bigwedge^*\mT^{1,0}M),\Delta,\dbar\right)$
be a compatible dGBV superalgebra. Then
$\left(\Omega^{0,*}(\bigwedge^*\mT^{1,0}M),\tilde{\Delta},\dbar\right)$
is a strongly compatible dGBV superalgebra.
\end{Lem}

\begin{proof}
The statement will be clear from the following three properties of the
operator $\tilde{\Delta}$ that we shall establish subsequently:
It is compatible with the Schouten-Nijenhuis bracket, anticommutes with $\dbar$ and
squares to zero.

Let $\alpha\in\Omega^{0,r}(\bigwedge^s\mT^{1,0}M)$ and $\beta\in\Omega^{0,q}(\bigwedge^p\mT^{1,0}M)$. Applying $\Pi_{s+p-1,r+q}$ to
both sides of the equation of Schouten-Nijenhuis compatibility with respect to
$\Delta$, we obtain
\begin{align*}
\scal[[]{\alpha}{\beta}=(-1)^{\deg(\alpha)}\tilde{\Delta}(\alpha\cdot\beta)
-(-1)^{\deg(\alpha)}\Pi_{s+p-1,r+q}(\Delta(\alpha)\cdot\beta)-\Pi_{s+p-1,r+q}(\alpha\cdot\Delta(\beta))
\end{align*}
from which compatibility of $\tilde{\Delta}$ is clear.

The equation $\dbar\circ\tilde{\Delta}+\tilde{\Delta}\circ\dbar=0$ follows from
the corresponding property of $\Delta$ by an analogous calculation.

It remains to show $\tilde{\Delta}^2=0$. Obviously, this is true upon
application to $\alpha\in\Omega^{0,r}M=\Omega^{0,r}(\bigwedge^0\mT^{1,0}M)$
since, for such $\alpha$, already $\tilde{\Delta}(\alpha)=0$ holds by
definition of $\tilde{\Delta}$. For $\alpha\in\Omega^{0,r}(\bigwedge^1\mT^{1,0}M)$,
we have that $\tilde{\Delta}(\alpha)\in\Omega^{0,r}M$ and, therefore,
also $\tilde{\Delta}^2(\alpha)=0$.
By induction on $s$, we now assume that $\tilde{\Delta}^2(\alpha)=0$
for $\alpha\in\Omega^{0,r}(\bigwedge^{s-1}\mT^{1,0}M)$ with $r\in\bN$ arbitrary.
Let $\alpha\in\Omega^{0,r}(\bigwedge^s\mT^{1,0}M)$. Locally we may write
$\alpha=\beta\cdot\gamma$ with $\beta\in\Omega^{0,r}(\bigwedge^1\mT^{1,0}M)$
and $\gamma\in\Omega^{0,0}(\bigwedge^{s-1}\mT^{1,0}M)$. Using Schouten-Nijenhuis compatibility,
we calculate
\begin{align*}
\tilde{\Delta}^2(\alpha)&=\tilde{\Delta}\left((-1)^{\deg(\beta)}\scal[[]{\beta}{\gamma}+\tilde{\Delta}(\beta)\cdot\gamma+(-1)^{\deg(\beta)}\beta\cdot\tilde{\Delta}(\gamma)\right)\\
&=(-1)^{\deg(\beta)}\tilde{\Delta}\scal[[]{\beta}{\gamma}+(-1)^{\deg(\beta)-1}\scal[[]{\tilde{\Delta}(\beta)}{\gamma}+(-1)^{\deg(\beta)-1}\tilde{\Delta}(\beta)\cdot\tilde{\Delta}(\gamma)\\
&\qquad+\scal[[]{\beta}{\tilde{\Delta}(\gamma)}+(-1)^{\deg(\beta)}\tilde{\Delta}(\beta)\cdot\tilde{\Delta}(\gamma)\\
&=(-1)^{\deg(\beta)}\tilde{\Delta}\scal[[]{\beta}{\gamma}+(-1)^{\deg(\beta)-1}\scal[[]{\tilde{\Delta}(\beta)}{\gamma}+\scal[[]{\beta}{\tilde{\Delta}(\gamma)}\\
&=\Pi_{s-1,r}\left((-1)^{\deg(\beta)}\Delta\scal[[]{\beta}{\gamma}+(-1)^{\deg(\beta)-1}\scal[[]{\Delta(\beta)}{\gamma}+\scal[[]{\beta}{\Delta(\gamma)}\right)
\end{align*}
Now, performing the first two steps in this calculation in reverse order,
with $\tilde{\Delta}$ replaced by $\Delta$, we find
\begin{align*}
\tilde{\Delta}^2(\alpha)=\Pi_{s-1,r}\circ\Delta^2(\alpha)=0
\end{align*}
which was to be shown.
\end{proof}

\section{A Generalised Tian-Todorov Formula}
\label{secTianTodorov}

This section establishes one direction of Thm. \ref{thmWeakCYBV} by means of
a generalised Tian-Todorov lemma.
We shall use a suitable generalisation of
the classical operator $\partial:\Omega^{p,q}\rightarrow\Omega^{p+1,q}$
(with $q=0$) on a complex manifold. This operator acts on integral
forms $\partial:I^{n-p}\rightarrow I^{n-p+1}$ on $M$
rather than on differential forms and, modulo some
identifications, was already studied in \cite{Man88}.
We refer to the appendix, Sec. \ref{secCanonicalBundle}, for details,
and be very brief at this point.
A local formula, in terms of coordinates $(\xi^k)$, reads as follows.
\begin{align*}
&\partial\left(f\cdot\dd{\xi^1}\wedge\ldots\wedge\dd{\xi^p}\otimes[d\xi]\right)\\
&\qquad\qquad:=\sum_{i=1}^{n+m}(-1)^{M_i}\dd[f]{\xi^i}\cdot\left(\dd{\xi^1}\wedge\ldots\wedge\dd{\xi^{i-1}}\wedge\widehat{\dd{\xi^i}}\wedge\dd{\xi^{i+1}}\wedge\ldots\wedge\dd{\xi^p}\right)\otimes[d\xi]
\end{align*}
where $M_i$ is the sign arising from moving $\dd{\xi^i}$ to the front.
We shall denote the extension (\ref{eqnPartial}) to an operator
\begin{align*}
\partial:\Omega^{0,q}\left(I^{n-p}\right)=\Omega^{0,q}M\otimes I^{n-p}\longrightarrow\Omega^{0,q}M\otimes I^{n-p+1}=\Omega^{0,q}\left(I^{n-p+1}\right)
\end{align*}
by the same symbol.

For the rest of this section, we assume that $\Ber M$ is trivial,
and we fix a global trivialising section $\omega\in\Ber M$.
Every $\alpha\in\Ber M$ is then uniquely determined by a unique superfunction
$f\in\mO_M$ such that $\alpha=f\cdot\omega$.
Without loss of generality, $\omega$ may (and will) be chosen homogeneous, that is
of parity $\abs{\omega}=(\dim M)_{\overline{1}}$.
Locally, with respect to coordinates $(\xi^k)$, we can write $\omega=h\cdot[d\xi]$
with $h$ a local, even, invertible and holomorphic superfunction.
Moreover, a choice of $\omega$ defines an isomorphism (even or odd)
\begin{align*}
\eta:\bigwedge^p\mT^{1,0}M\rightarrow I^{n-p}\;,\qquad\eta(v_1\wedge\ldots\wedge v_p):=(v_1\wedge\ldots\wedge v_p)\otimes\omega
\end{align*}
The map $\eta$ induces canonical isomorphisms, denoted by the same symbol:
\begin{align*}
\eta:\Omega^{0,q}\left(\bigwedge^p\mT^{1,0}M\right)
\longrightarrow\Omega^{0,q}\left(I^{n-p}\right)
\end{align*}
On both sides, there is the $\dbar$-operator (\ref{eqnDbar}). By construction,
it acts only on the form part, whence we obtain
\begin{align}
\label{eqnEtaDbar}
\eta\circ\dbar=\dbar\circ\eta
\end{align}

\begin{Def}
\label{defDelta}
We define the operator $\Delta^{\omega}$ as follows.
\begin{align*}
\Delta^{\omega}:\Omega^{0,q}\left(\bigwedge^p\mT^{1,0}M\right)\rightarrow\Omega^{0,q}\left(\bigwedge^{p-1}\mT^{1,0}M\right)
\;,\qquad\Delta^{\omega}:=\eta^{-1}\circ\partial\circ\eta
\end{align*}
\end{Def}

\begin{Lem}
\label{lemDeltaOmega}
The operator $\Delta^{\omega}$ satisfies the following properties,
for every open subset $U\subseteq M_{\overline{0}}$.
It vanishes on functions $f\in\mO_M(U)$ and forms $\lambda\in\Omega^{0,1}M(U)$.
Applied to vectors $X\in\mT^{1,0}M(U)$, it takes values in functions
$\Delta^{\omega}(X)\in\mO_M(U)$.
The local formula reads $\Delta^{\omega}(\partial_{\xi^k})=\dd[h]{\xi^k}h^{-1}$,
where $\omega=h\cdot[d\xi]$.
\end{Lem}

\begin{proof}
This statement is clear by the definition of $\partial$ (and $\eta$),
with the calculation for coordinate vector fields as follows.
\begin{align*}
\Delta^{\omega}(\partial_{\xi^k})
=\eta^{-1}\partial\left(h\cdot\partial_{\xi^k}\otimes[d\xi]\right)
=\eta^{-1}\left(\dd[h]{\xi^k}\cdot[d\xi]\right)
=\dd[h]{\xi^k}h^{-1}
\end{align*}
\end{proof}

The identity of the following result was stated in \cite{BK98}
for the case of a classical complex manifold with trivial canonical bundle.
This, in turn, generalises a lemma proved independently by both
Tian \cite{Tia87} and Todorov \cite{Tod89}.
We will, therefore, refer to the following supergeometric generalisation
also as ''generalised Tian-Todorov lemma''.

\begin{Lem}[Tian-Todorov]
\label{lemTianTodorov}
Let $\alpha\in\Omega^{0,p}(\bigwedge^s\mT^{1,0}M)$ and $\beta\in\Omega^{0,q}(\bigwedge^r\mT^{1,0}M)$. Then the operator $\Delta^{\omega}$ just defined is
compatible with the Schouten-Nijenhuis bracket $\scal[[]{\cdot}{\cdot}$
in the following sense.
\begin{align*}
-\scal[[]{\alpha}{\beta}=(-1)^{\deg(\alpha)}\Delta^{\omega}(\alpha\wedge\beta)-(-1)^{\deg(\alpha)}\Delta^{\omega}(\alpha)\wedge\beta
-\alpha\wedge\Delta^{\omega}(\beta)
\end{align*}
\end{Lem}

The Schouten-Nijenhuis bracket in the present context is defined in
Def. \ref{defSchoutenBracket} below, while
$\deg(\alpha)=p+s$ denotes the cohomological degree of $\alpha$ as in
Sec. \ref{secBV}.

\begin{proof}
The local formula of Lem. \ref{lemDeltaOmega}
for $\Delta$ applied to tangent vectors generalises as follows.
\begin{align*}
&\Delta(f\cdot\partial_{\xi^1}\wedge\ldots\wedge\partial_{\xi^r})\\
&\qquad=\sum_i(-1)^{(i-1)+\abs{\xi^i}(\abs{f}+\sum_{l=1}^{i-1}\abs{\xi^l})}\dd[(f\cdot h)]{\xi^i}h^{-1}(\partial_{\xi^1}\wedge\ldots\wedge\widehat{\partial_{\xi^i}}\wedge\ldots\wedge\partial_{\xi^r})
\end{align*}
As usual, the hat symbol means omission.
The Tian-Todorov formula, for the case $p=q=0$, is established from
this formula and the definition of the Schouten-Nijenhuis bracket,
Def. \ref{defSchoutenBracket}, in a lengthy but straightforward calculation.

In the second step, we observe that the general case can be deduced from
the case $p=q=0$ already established. Again, this is a direct calculation
in local coordinates, unwinding the definitions.
\end{proof}

\begin{Prp}
\label{prpBV}
The triple $\left(\Omega^{0,*}(\bigwedge^*\mT^{1,0}M),\Delta^{\omega},\dbar\right)$
is a dGBV superalgebra strongly compatible with the Schouten-Nijenhuis bracket.
\end{Prp}

We remark that, thanks to this result together with
Lem. \ref{lemSchoutenBVGenerators}, the operator $\Delta^{\omega}$
is completely determined by its values stated in Lem. \ref{lemDeltaOmega}.

\begin{proof}
By construction, $\Delta^{\omega}$ is a deg-odd linear map.
$(A,\Delta^{\omega})$ is a BV superalgebra by Lem. \ref{lemSchoutenDerivation},
together with the Tian-Todorov lemma, Lem. \ref{lemTianTodorov}.
Compatibility holds true by the same formula.
It remains to show that, together with $\dbar$, it is in fact dGBV.
Strong compatibility is then clear by construction of $\Delta^{\omega}$.
In fact, by (\ref{eqnPartialSquared}),
we see that $\Delta^2=\eta^{-1}\circ\partial^2\circ\eta=0$.
Moreover, the formula $\dbar\circ\Delta^{\omega}+\Delta^{\omega}\circ\dbar=0$
holds true since $\dbar$ commutes with $\eta$ and anticommutes with
$\partial$, by (\ref{eqnEtaDbar}) and Lem. \ref{lemPartialDbar}, respectively.
\end{proof}

\section{Flat Connections on the Canonical Bundle}
\label{secFlatConnections}

Having successfully translated triviality of the canonical bundle into
existence of a dGBV structure, we now turn to the converse, thus establishing
the remaining direction of Thm. \ref{thmWeakCYBV}.
In the following, we assume that $M$ is a complex supermanifold equipped
with an operator
$\Delta:\Omega^{0,*}(\bigwedge^*\mT^{1,0}M)\circlearrowleft$, such that
the triple
$\left(\Omega^{0,*}(\bigwedge^*\mT^{1,0}M),\Delta,\dbar\right)$
is a dGBV superalgebra strongly compatible with the Schouten-Nijenhuis bracket.
We shall construct a flat connection on $\Ber M$ and establish triviality
by means of a parallel section obtained by superholonomy theory,
which is sketched in the appendix.
To be precise, we shall need the Holonomy Principle,
Thm. \ref{thmHolonomyPrinciple}, together with the Ambrose-Singer Theorem,
cf. (\ref{eqnHolonomyAlgebra}).
At this point, the assumption that $M$ be simply-connected becomes important.

\begin{Lem}
\label{lemBVConnection}
Let $(\xi^k)$ denote complex coordinates of $M$ with induced real
coordinates (\ref{eqnComplexCoordinates}). Then the prescription
\begin{align*}
(\nabla^{\xi})_{\partial_{\xi^k_R}}[d\xi]:=-\Delta(\partial_{\xi^k})\cdot[d\xi]
\;,\qquad
(\nabla^{\xi})_{\partial_{\xi^k_I}}[d\xi]:=-i\cdot\Delta(\partial_{\xi^k})\cdot[d\xi]
\end{align*}
together with the Leibniz rule defines a connection on $\Ber M$.
\end{Lem}

We remark that, as usual, this definition is meant to be done with respect
to any coordinate system, and the resulting expressions are claimed to
transform such as to constitute a global object.
By construction, this is then a connection.
Moreover, we remind the reader of our convention that a connection is
always real, with complex linearity if present considered as an extra structure.
In the case at hand, we find
\begin{align*}
(\nabla^{\xi})_{JX}[d\xi]=i\cdot(\nabla^{\xi})_X[d\xi]
\end{align*}
which, upon complex linear extension, implies
\begin{align*}
(\nabla^{\xi})_{\partial_{\xi^k}}[d\xi]
=(\nabla^{\xi})_{\partial_{\xi^k_R}}[d\xi]=-\Delta(\partial_{\xi^k})\cdot[d\xi]
\end{align*}

\begin{proof}
By the preceding remark, we may work completely in the complex picture.
We need to show the following analogon of (\ref{eqnConnectionPullback}) under
a coordinate transformation $\varphi:\zeta\rightarrow\xi$.
\begin{align*}
\varphi^{\star}((\nabla^{\xi})_{\partial_{\xi^k}}[d\xi])
=(\nabla^{\zeta})_{\varphi^{\star}\partial_{\xi^k}}\varphi^{\star}[d\xi]
\end{align*}
We calculate the right hand side, using (\ref{eqnCoordinateVFTrafo}) and Leibniz' rule.
\begin{align*}
(\nabla^{\zeta})_{\varphi^{\star}\partial_{\xi^k}}\varphi^{\star}[d\xi]
&=(-1)^{(m+k)m}(d\varphi^{-1})^m_{\phantom{m}k}\partial_{\zeta^m}(\sdet d\varphi)\cdot[d\zeta]\\
&\qquad-(-1)^{(m+k)m}(d\varphi^{-1})^m_{\phantom{m}k}\sdet d\varphi\cdot\Delta(\partial_{\zeta^m})\cdot[d\zeta]\\
&=\partial_{\zeta^m}\left((d\varphi^{-1})^m_{\phantom{m}k}\sdet d\varphi\right)[d\zeta]
-\partial_{\zeta^m}(d\varphi^{-1})^m_{\phantom{m}k}\sdet d\varphi\cdot[d\zeta]\\
&\qquad-(-1)^{(m+k)m}(d\varphi^{-1})^m_{\phantom{m}k}\sdet d\varphi\cdot\Delta(\partial_{\zeta^m})\cdot[d\zeta]
\end{align*}
By Lem. \ref{lemJacobiSum}, the first term vanishes.
Further utilising the compatibility equation with the Schouten-Nijenhuis bracket,
we further calculate
\begin{align*}
(\nabla^{\zeta})_{\varphi^{\star}\partial_{\xi^k}}\varphi^{\star}[d\xi]&=-\left(\partial_{\zeta^m}(d\varphi^{-1})^m_{\phantom{m}k}+\Delta(\partial_{\zeta^m})\cdot(d\varphi^{-1})^m_{\phantom{m}k}\right)\sdet d\varphi\cdot[d\zeta]\\
&=-\Delta(\partial_{\zeta^m}\cdot(d\varphi^{-1})^m_{\phantom{m}k})\sdet d\varphi\cdot[d\zeta]\\
&=-\Delta(\varphi^{\sharp}\circ\partial_{\xi^k}\circ(\varphi^{-1})^{\sharp})\varphi^{\star}[d\xi]\\
&=-\varphi^*\Delta(\partial_{\xi^k})\varphi^{\star}[d\xi]\\
&=\varphi^{\star}((\nabla^{\xi})_{\partial_{\xi^k}}[d\xi])
\end{align*}
thus proving well-definedness.
\end{proof}

\begin{Lem}
\label{lemBVConnectionFlat}
The connection from Lem. \ref{lemBVConnection} is flat (has vanishing curvature).
\end{Lem}

\begin{proof}
We calculate
\begin{align*}
\scal[R]{\partial_{\xi^l}}{\partial_{\xi^m}}[d\xi]&=\nabla_{\partial_{\xi^l}}\nabla_{\partial_{\xi^m}}[d\xi]-(-1)^{lm}\nabla_{\partial_{\xi^m}}\nabla_{\partial_{\xi^l}}[d\xi]\\
&=-\nabla_{\partial_{\xi^l}}\left(\Delta(\partial_{\xi^m})\cdot[d\xi]\right)+(-1)^{lm}\nabla_{\partial_{\xi^m}}\left(\Delta(\partial_{\xi^l})[d\xi]\right)\\
&=\left(-\partial_{\xi^l}(\Delta(\partial_{\xi^m}))+(-1)^{lm}\partial_{\xi^m}(\Delta(\partial_{\xi^l}))\right)[d\xi]\\
&\qquad+\left(-(-1)^{lm}\Delta(\partial_{\xi^m})\nabla_{\partial_{\xi^l}}[d\xi]+\Delta(\partial_{\xi^l})\nabla_{\partial_{\xi^m}}[d\xi]\right)\\
&=:(1)+(2)
\end{align*}
where
\begin{align*}
(2)=\left((-1)^{lm}\Delta(\partial_{\xi^m})\Delta(\partial_{\xi^l})-\Delta(\partial_{\xi^l})\Delta(\partial_{\xi^m})\right)[d\xi]=0
\end{align*}
Consider the term
\begin{align*}
(-1)^{lm}\partial_{\xi^m}(\Delta(\partial_{\xi^l}))
=(-1)^{lm}\scal[[]{\partial_{\xi^m}}{\Delta(\partial_{\xi^l})}
=-\scal[[]{\Delta(\partial_{\xi^l})}{\partial_{\xi^m}}
\end{align*}
By Lem. \ref{lemGBVCompatible}, it equals
\begin{align*}
(-1)^{lm}\partial_{\xi^m}(\Delta(\partial_{\xi^l}))
=\scal[[]{\partial_{\xi^l}}{\Delta(\partial_{\xi^m})}
=\partial_{\xi^l}(\Delta(\partial_{\xi^m}))
\end{align*}
whence also $(1)=0$ vanishes. The statement is proved.
\end{proof}

\begin{Prp}
\label{prpBVParallelSection}
Let $M$ be simply connected. Under the above hypotheses (existence
of a strongly compatible dGBV structure),
there exists a global parallel nonzero homogeneous and holomorphic
section $\omega\in\Ber M$, which is unique up to multiplication by
a complex number.
\end{Prp}

\begin{proof}
By Lem. \ref{lemBVConnectionFlat},
the connection $\nabla$ from Lem. \ref{lemBVConnection} is flat.
Under the hypothesis of simply connectedness, the Ambrose-Singer Theorem
for the superholonomy functor (cf. (\ref{eqnHolonomyAlgebra}) in the
appendix) implies that $\Hol^{\nabla}_T=\{1\}$.
By the holonomy principle, Thm. \ref{thmHolonomyPrinciple},
there exists a global parallel nonzero section $\omega\in\Ber M$, which
is unique upon specifying $\omega_x\in x^*\Ber M\cong\bC$ for
some (topological) point $x\in M_{\overline{0}}$
(in the framework of \cite{Gro14c}, we consider $S=\bR^{0|0}$).
In terms of coordinates we may, locally, write $\omega=h\cdot[d\xi]$.
By construction, the superfunction $h$ is invertible and even.

It remains to show that $h$ is holomorphic.
As $\omega$ is parallel, we find
\begin{align*}
0=\nabla_{\partial_{\xi^k}}\omega=\nabla_{\partial_{\xi^k}}(h[d\xi])
=\partial_{\xi^k}(h)[d\xi]+h\nabla_{\partial_{\xi^k}}[d\xi]
=\left(\partial_{\xi^k}(h)-h\Delta(\partial_{\xi^k})\right)[d\xi]
\end{align*}
Therefore
\begin{align}
\label{eqnDeltaFormula}
h\Delta(\partial_{\xi^k})=\partial_{\xi^k}(h)
\end{align}
Since $\Delta$ is, by assumption, complex linear, it follows
thet $h$ and, therefore, $\omega$ is holomorphic.
\end{proof}

By means of the trivialising section $\omega\in\Ber M$ provided by
Prp. \ref{prpBVParallelSection}, we may use Def. \ref{defDelta} to
define an operator, to be denoted $\Delta^{\omega}$.
From the previous section, we know that $\Delta^{\omega}$ constitutes
a strongly compatible dGBV structure, just like $\Delta$.
By the following statement, both structures coincide.

\begin{Lem}
\label{lemDeltaDelta}
Let $\omega\in\Ber M$ be induced by $\Delta$ as in
Prp. \ref{prpBVParallelSection}. Then $\Delta=\Delta^{\omega}$.
\end{Lem}

\begin{proof}
By (\ref{eqnDeltaFormula}), we obtain
$\Delta(\partial_{\xi^k})=h^{-1}\partial_{\xi^k}(h)$.
The same formula holds with $\Delta$ replaced with $\Delta^{\omega}$,
by Lem. \ref{lemDeltaOmega}. Since both operators are, moreover,
strongly compatible, they must agree (cf. Lem. \ref{lemSchoutenBVGenerators}).
\end{proof}

\begin{proof}[Proof of Thm. \ref{thmWeakCYBV}]
Assuming that $\Ber M$ is trivial, we let $\omega\in\Ber M$ denote a
homogeneous trivialising global section.
The operator $\Delta^{\omega}$ from Def. \ref{defDelta} then induces
a strongly compatible dGBV structure (cf. Prp. \ref{prpBV}).
We thus obtain a map
$F:\omega\mapsto\Delta^{\omega}$.

On the other hand, let $\Delta$ constitute a strongly compatible dGBV structure.
By Prp. \ref{prpBVParallelSection}, we obtain a homogeneous trivialising section
$\omega_{\Delta}\in\Ber M$
(up to a constant) which is parallel with respect to the connection
$\nabla^{\Delta}$ of Lem. \ref{lemBVConnection}.
We denote this map by $G:\Delta\mapsto\omega_{\Delta}$
(defined up to a constant).
By Lem. \ref{lemDeltaDelta}, we get $F\circ G=\id_{\Delta}$.

It remains to show $G\circ F=\id_{\omega}$ (up to a constant).
Let $\omega\in\Ber M$ be trivialising and homogeneous.
With respect to coordinates $(\xi^k)$ we write, locally, $\omega=h\cdot[d\xi]$.
Moreover, we set $\tilde{\omega}:=G\circ F(\omega)$ and, locally,
$\tilde{\omega}=\tilde{h}\cdot[d\xi]=f\cdot h\cdot[d\xi]$.
Using Lem. \ref{lemDeltaOmega} and (\ref{eqnDeltaFormula}), we obtain
\begin{align*}
h^{-1}\partial_{\xi^k}(h)=\Delta^{\omega}(\partial_{\xi^k})
=\tilde{h}^{-1}\partial_{\xi^k}(\tilde{h})
=h^{-1}\partial_{\xi^k}(h)+f^{-1}\partial_{\xi^k}(f)
\end{align*}
It follows that $f$ is constant, which was to be shown.
\end{proof}

\section{The Calabi-Yau Case}
\label{secCY}

In this section, we study a particular kind of supermanifolds with
trivial canonical bundle that we refer to as Calabi-Yau.
To begin with, consider the following superisation of one of the equivalent
characterisations of a K\"ahler manifold.

\begin{Def}
\label{defKaehler}
A \emph{K\"ahler supermanifold} is a semi-Riemannian supermanifold
$(M,g)$ such that the holonomy group functor (associated with the
Levi-Civita connection) satisfies
\begin{align*}
\Hol_x(T)\subseteq U_{p_0,q_0|p_1,q_1}(\mO_T)
\end{align*}
for all $T=\bR^{0|L'}$ at some $x\in M_{\overline{0}}$.
\end{Def}

Here, $(p_0,q_0)$ and $(p_1,q_1)$ denote the signatures of the
underlying metrics in the even/even and odd/odd directions, respectively,
in the definition of the unitary supergroup.
It is clear that Def. \ref{defKaehler} is independent of $x$
(we remind the reader that we consider only connected supermanifolds).
Considering a different point $y$ results in the holonomy conjugated by parallel
transport from $x$ to $y$.
By the Twofold Theorem of \cite{Gro16}, this definition
is equivalent to Galaev's given in \cite{Gal09}.
By the Holonomy Principle, Thm. \ref{thmHolonomyPrinciple},
a K\"ahler supermanifold can be
characterised as a Hermitian supermanifold $(M,g)$ with a parallel
complex structure $\nabla J=0$.

The Calabi-Yau case is not so straightforward.
Already for classical manifolds, that term is used to denote
several different concepts, which are tightly related but not equivalent.
This is summarised e.g. in Sec. 6.1 of \cite{Joy00}.
The various different definitions of a Calabi-Yau manifolds can
all be promoted to supergeometry, but these generalisations have even less
in common, as we shall presently see.
To be definite, we propose the following definition which, from the point of
view of holonomy theory, is the most natural one.

\begin{Def}
\label{defCalabiYau}
A \emph{Calabi-Yau supermanifold} is a K\"ahler supermanifold
such that
\begin{align*}
\Hol_x(T)\subseteq SU_{p_0,q_0|p_1,q_1}(\mO_T)
\end{align*}
for all $T=\bR^{0|L'}$ at some $x\in M_{\overline{0}}$.
\end{Def}

This definition may be translated to Galaev's holonomy theory
by means of the aforementioned Twofold Theorem.
A Calabi-Yau supermanifold
in our denomination is then equivalent to
what is called ''special K\"ahler supermanifold'' in \cite{Gal09}.
We state the following characterisation in the simply-connected case,
which was proved as Prp. 11.1 in that reference.

\begin{Prp}
\label{prpCalabiYau}
Let $(M,g)$ be a simply-connected K\"ahler supermanifold. Then
$(M,g)$ is Calabi-Yau if and only if the Ricci tensor vanishes.
\end{Prp}

Def. \ref{defCalabiYau} is the natural generalisation of the
classical definition of a Calabi-Yau manifold in terms of holonomy,
and Prp. \ref{prpCalabiYau} generalises a well-known classical characterisation.
However, the analogy to classical geometry is not as close for
the following two reasons.

First, a Calabi-Yau supermanifold is, in particular, a K\"ahler supermanifold.
Moreover, the smooth manifold underlying a K\"ahler supermanifold
naturally inherits a K\"ahler structure.
However, the classical K\"ahler manifold underlying a Calabi-Yau supermanifold
need not be a Calabi-Yau manifold.

Second, Calabi-Yau manifolds obtained their name from Calabi's conjecture
which was later proved by Yau.
The analogous theorem in supergeometry is, however, wrong.
There are explicit counterexamples, see \cite{RW05}.
Nevertheless, the following generalisation of a well-known result
holds true.

\begin{Prp}
\label{prpCYTrivialBer}
The Levi-Civita connection of a Calabi-Yau supermanifold induces
a connection $\nabla^{\Ber}$ on $\Ber M$ with trivial holonomy.
There is a global parallel nowhere vanishing homogeneous and holomorphic section
$\omega\in\Ber M$, unique up to a complex number.
In particular, a Calabi-Yau supermanifold has trivial canonical bundle.
\end{Prp}

It follows at once (by Prp. \ref{prpBV}) that a Calabi-Yau supermanifold bears
a strongly-compatible dGBV structure.
Prp. \ref{prpCYTrivialBer} will be established with the help of the
following two lemmas. In the first, we construct the connection
$\nabla^{\Ber}$, while the second is concerned about its parallel transports.

\begin{Lem}
\label{lemCYBerConnection}
A connection $\nabla$ on the tangent bundle of a complex supermanifold,
which satisfies $\nabla J=0$,
induces a connection on $\Ber M$, denoted $\nabla^{\Ber}$, through the
following local definition in coordinates $(\xi^k)$.
\begin{align*}
\nabla^{\Ber}_{\partial_{\xi^l}}[d\xi]:=-\str\left(\jmath\circ(\nabla^{\xi})_{\dd{\xi^l}}\circ\jmath^{-1}\right)\cdot[d\xi]
\end{align*}
\end{Lem}

Here, conjugation with $\jmath$ translates $\nabla$ to a connection
(\ref{eqnJmathNabla}) on $\mT^{1,0}M$, which the definition requires
to be complex linear.
The condition $\nabla J=0$ is necessary and sufficient for that property.

\begin{proof}
We need to show that the local definition given behaves correctly under a
coordinate transformation $\varphi:\zeta\rightarrow\xi$.
This proof largely parallels that of Lem. \ref{lemBVConnection}.
Abbreviating $\nabla:=\nabla^{\Ber}$,
one side of the equation to be established may be expressed as
\begin{align*}
(\nabla^{\zeta})_{\varphi^{\star}\partial_{\xi^k}}\varphi^{\star}[d\xi]
=-\left(\partial_{\zeta^m}(d\varphi^{-1})^m_{\phantom{m}k}+\str\left(\jmath\circ(\nabla^{\zeta})_{\partial_{\zeta^m}}\circ\jmath^{-1}\right)\cdot(d\varphi^{-1})^m_{\phantom{m}k}\right)\sdet d\varphi\cdot[d\zeta]
\end{align*}
whereas the other side reads as follows.
\begin{align*}
\varphi^{\star}((\nabla^{\xi})_{\partial_{\xi^k}}[d\xi])
=-\str(\jmath\circ\varphi^{\sharp}(^R(\Gamma^{\xi})_{k\cdot}^{\cdot})\circ\jmath^{-1})\cdot\varphi^{\star}[d\xi]
\end{align*}
The Christoffel symbols transform according to Lem. \ref{lemChristoffelTrafo}.
Further utilising Jacobi's formula (\ref{eqnJacobiFormula}), together with
Lem. \ref{lemJacobiSum}, one arrives at the same expression as before, thus
showing well-definedness.
\end{proof}

\begin{Lem}
\label{lemSdetParallelTransport}
Let $P_{\gamma}$ denote parallel transport with respect to the
connection $\nabla$ along some $S$-path $\gamma$. Then
\begin{align*}
P^{\Ber}_{\gamma}=\sdet(\jmath\circ P_{\gamma}\circ\jmath^{-1})^{-1}
\end{align*}
is the parallel transport with respect to $\nabla^{\Ber}$ along the
same path.
\end{Lem}

Parallel transport with respect to the induced connection on the cotangent
bundle $\mT^*M$
comes with an inverse. Given the construction of the Berezinian,
the formula stated is thus very natural.
In fact, the connection in Lem. \ref{lemCYBerConnection} was constructed
such as to have this parallel transport associated.
The sign in that definition can be understood
as the linearisation of the aforementioned inverse.

\begin{proof}
We must show that $P^{\Ber}_{\gamma}$ as stated satisfies
the differential equation for parallel transport in coordinates $(\xi^k)$,
(\ref{eqnParallelTransport}), which reads as follows.
\begin{align*}
\partial_tP^{\Ber}_{\gamma}
=-\partial_t(\gamma^*(\xi^l))\cdot\left(-\str(\jmath(^R(\Gamma^{\xi})_{l\cdot}^{\cdot}))\right)\cdot
P^{\Ber}_{\gamma}
\end{align*}
By Jacobi's formula (\ref{eqnJacobiFormula}), the left hand side of this
equation can be written as follows.
\begin{align*}
\partial_tP^{\Ber}_{\gamma}=-\str(\jmath\circ(\partial_tP_{\gamma})P_{\gamma}^{-1}\circ\jmath^{-1})\cdot P^{\Ber}_{\gamma}
\end{align*}
Replacing, in that expression, $\partial_tP_{\gamma}$ with the right hand side
of the parallel transport equation (\ref{eqnParallelTransport}) with respect to
$P_{\gamma}$, the analogous equation for $P^{\Ber}_{\gamma}$ follows.
\end{proof}

\begin{proof}[Proof of Prp. \ref{prpCYTrivialBer}]
By hypothesis, the parallel transport $P_{\gamma}$ with respect to the
Levi-Civita connection along some loop
$\gamma:x\rightarrow x$ is contained in some special unitary supergroup.
In particular, $\sdet(\jmath\circ P_{\gamma}\circ\jmath^{-1})=1$ holds.
By Lem. \ref{lemSdetParallelTransport}, the holonomy group functor
with respect to the induced connection $\nabla^{\Ber}$ on $\Ber M$
is trivial.
Continuing analogous to the proof of Prp. \ref{prpBVParallelSection},
one finds a global parallel nonzero homogeneous and holomorphic section
$\omega\in\Ber M$, unique up to a complex number.
\end{proof}

We are now in the position to give a more precise statement and
proof of our second main theorem, Thm. \ref{thmIntroCYConnections}.

\begin{Thm}
\label{thmCYConnections}
Let $M$ be a Calabi-Yau supermanifold and $\omega\in\Ber M$ be the
trivialising section of Prp. \ref{prpCYTrivialBer}, parallel with respect
to $\nabla^{\Ber}$. Let $\nabla^{\Delta}$ denote the connection from
Lem. \ref{lemBVConnection} with respect to the operator
$\Delta=\Delta^{\omega}$ defined in
Def. \ref{defDelta}. Then $\nabla^{\Ber}=\nabla^{\Delta}$ agree.
In particular, it holds
\begin{align*}
\Delta(\partial_{\xi^k})=\str\left(\jmath\circ\nabla_{\dd{\xi^k}}\circ\jmath^{-1}\right)
\end{align*}
\end{Thm}

\begin{proof}
By assumption, $\omega$ is parallel with respect to $\nabla^{\Ber}$.
By the second part of the proof of Thm. \ref{thmWeakCYBV},
it is also parallel with respect to $\nabla^{\Delta}$.
Writing $\omega=h\cdot[d\xi]$ (in coordinates $(\xi^k)$), we thus find
\begin{align*}
\left(\partial_{\xi^k}(h)-h\Delta(\partial_{\xi^k})\right)[d\xi]
=\nabla^{\Delta}_{\partial_{\xi^k}}\omega=0
=\nabla^{\Ber}_{\partial_{\xi^k}}\omega
=\left(\partial_{\xi^k}(h)-h\,\str\left(\jmath\circ\nabla_{\partial_{\xi^k}}\circ\jmath^{-1}\right)\right)[d\xi]
\end{align*}
from which the formula claimed is obvious. Moreover, the connections agree.
\end{proof}

\section*{Acknowledgements}

I would like to thank Anton Galaev for helpful discussions concerning
supergroups and Alexander Alldridge for making suggestions
which helped to improve a previous version of the article.

\appendix

\section{Elements of Complex Supergeometry}
\label{secComplexSupergeometry}

This appendix is meant to serve as a short introduction
to selected elements of complex supergeometry needed in the main text.
Although the theory presented here is certainly not new,
we consider it worth to have it collected in a uniform and consistent fashion.
The definitions, conventions and results contained for easy reference
allow for a more concise form of the main text with less disruptions,
as a reward for the few extra pages here.

The first subsection is rather general and
applies simultaneously to the smooth and complex categories.
This is different from the main text, and the remainder of this appendix.
Starting with the second subsection, topics include definition
and properties of the Schouten-Nijenhuis bracket,
parallel transport and superholonomy and, finally, an exposition on
the canonical bundle and integral forms.

\subsection{Supermanifolds and Vector Bundles}

Let $M=(M_{\overline{0}},\mO_M)$ be a supermanifold in the sense of
Berezin-Kostant-Leites (\cite{Lei80}, \cite{Var04}, \cite{CCF11})
with underlying classical manifold $M_{\overline{0}}$.
We allow $M$ to be either real or complex.
Ignorant of some abuse of notation, a (super) vector bundle $\mE$ on $M$
can be defined as a sheaf of locally free $\mO_M$-supermodules.
The definition implies that transition functions are smooth (holomorphic)
if $M$ is a real (complex) supermanifold.
Examples include the (super) tangent bundle $\mT M$ which, by definition,
is the sheaf of $\mO_M$-superderivations.

Let $\mE\rightarrow N$ be a super vector bundle over $N$ and $\varphi:M\rightarrow N$ be a morphism
of supermanifolds. The pullback of $\mE$ under $\varphi$ is defined as
\begin{align}
\label{eqnPullbackBundle}
\varphi^*\mE(U):=\mO_M(U)\otimes_{\varphi}(\varphi_0^*\mE)(U)\;,\qquad U\subseteq M_0\quad\mathrm{open}
\end{align}
Here, $\varphi_0^*\mE$ is the pullback of the sheaf $\mE$ under the continuous map $\varphi_0$
which, in terms of its sheaf space, is the bundle of stalks $\mE_{\varphi_0(x)}$ attached to $x\in M_0$.
(\ref{eqnPullbackBundle}) indeed yields a vector bundle $\varphi^*\mE$ on $M$ of rank $\rk\mE$.
For details, consult \cite{Han12} and \cite{Ten75}.
In the same context, there is a canonical notion of pullback
$\varphi_0^*X\in\varphi_0^*\mE$ for a section $X\in\mE$, and the extension
\begin{align*}
\varphi^*X:=1\otimes_{\varphi}\varphi_0^*X\in\varphi^*\mE
\end{align*}
The definition is such that
\begin{align*}
\varphi^*(X\cdot f)=\varphi^*X\cdot\varphi^{\sharp}(f)\qquad\mathrm{for}\quad X\in\mE\;,\quad f\in\mO_N
\end{align*}
A local frame $(T^k)$ of $\mE$ gives rise to a local frame $(\varphi_0^*T^k)$ of $\varphi_0^*\mE$ and a local frame $(\varphi^*T^k)$ of $\varphi^*\mE$ such that,
locally, every section $X\in\varphi^*\mE$ can be written $X=\varphi^*T^k\cdot X^k$ with $X^k\in\mO_M(U)$.

In the case of the tangent bundle $\mT N$, there is a canonical
identification of the pullback $\varphi^*\mT N$ with the sheaf of derivations
along $\varphi$, through the prescription
\begin{align*}
\varphi^*X\mapsto\varphi^{\sharp}\circ X
\end{align*}
with the right hand side acting on functions $f\in\mO_N$.
Similarly, the pullback $\varphi^*\mT^*N$ of the cotangent bundle is
identified with the dual $(\varphi^*\mT N)^*$ through
\begin{align*}
\varphi^*\xi\mapsto\left(\varphi^*X\mapsto(\varphi^*\xi)(\varphi^*X):=\varphi^{\sharp}(\xi(X))\right)\qquad\mathrm{for}\quad\xi\in\mT^*N
\end{align*}
In the following, we shall use the aforementioned identifications
without an explicit mention.

The differential of $\varphi$ is defined by
\begin{align*}
d\varphi:\mT M\rightarrow\varphi^*\mT N\;,\qquad d\varphi[X]:=X\circ\varphi^{\sharp}
\end{align*}
In the case that $\varphi$ is an isomorphism, there is a further
identification of the pullback of a vector field $X\in\mT N$ or covector field $\xi\in\mT^*N$, respectively, as follows.
\begin{align}
\label{eqnIsoPullback}
\varphi^*X\mapsto\varphi^{\star}X:=\varphi^{\sharp}\circ X\circ(\varphi^{-1})^{\sharp}\in\mT M\;,\qquad\varphi^*\xi\mapsto\varphi^{\star}\xi:=(\varphi^*\xi)\circ d\varphi\in\mT^*M
\end{align}
We state the following natural properties.
\begin{align*}
(\varphi^{\star}\xi)(\varphi^{\star}X)=\varphi^{\sharp}(\xi(X))\;,\qquad
\varphi^{\star}df=d(\varphi^{\sharp}(f))\quad\mathrm{for}\qquad f\in\mO_N
\end{align*}

An isomorphism $\varphi:M\supseteq U\mapsto\bR^{n|m}$ is determined by
coordinates, that is
superfunctions, $(\xi^k)=(x^1,\ldots,x^n,\theta^1,\ldots,\theta^m)$ on $U$,
through the pullback of the corresponding functions on $\bR^{n|m}$.
Given coordinates, (co)vector fields on $U$ and (co)vector fields on $\bR^{n|m}$
are identified through (\ref{eqnIsoPullback}).
A vector field $X\in\mT M$ has a local expression
$X=\dd{\xi^l}\cdot X^l$ where $\dd{\xi^k}$ corresponds to the canonical
derivation on $\bR^{n|m}$ and $X^l$ is a uniquely determined superfunction on $U$.
With $(\zeta^i)$ denoting coordinates on $N$, the differential of
a map $\varphi:M\rightarrow N$ has the following local expression.
\begin{align*}
d\varphi^i_{\phantom{i}k}:=(-1)^{(\abs{\xi^k}+\abs{\zeta^i})\cdot\abs{\zeta^i}}\dd[\varphi^{\sharp}(\zeta^i)]{\xi^k}\quad
\mathrm{s.th.}\quad d\varphi[X]=\sum_{i,k}\left(\varphi^{\sharp}\circ\dd{\zeta^i}\right)\cdot d\varphi^i_{\phantom{i}k}\cdot X^k
\end{align*}
Interchanging derivatives gives rise to the following formula.
\begin{align}
\label{eqnDifferentialHessian}
\partial_j(d\varphi^m_{\phantom{m}n})
=(-1)^{\abs{n}\abs{j}+\abs{n}\abs{m}+\abs{j}\abs{m}}\partial_n(d\varphi^m_{\phantom{m}j})
\end{align}
We state the chain rule next.
Let $\varphi:M\rightarrow N$ and $\psi:N\rightarrow P$ be morphisms. Then
\begin{align}
d(\psi\circ\varphi)[X]
=\left(\varphi^{\sharp}\circ\psi^{\sharp}\circ\dd{\pi^l}\right)\cdot\varphi^{\sharp}(d\psi^l_{\phantom{l}i})\cdot d\varphi^i_{\phantom{i}k}\cdot X^k
\end{align}
with $(\pi^l)$ coordinates on $P$ and indices $k$, $i$ referring to (unlabelled)
coordinates on $M$ and $N$, respectively.

Let $\varphi$ be invertible. We write $d\varphi^{-1}:=(d\varphi)^{-1}$.
In terms of coordinates on $M$ and $N$, this is related to $d(\varphi^{-1})$ by
\begin{align}
\label{eqnDifferentialInverse}
(d\varphi^{-1})^l_{\phantom{l}m}
=\varphi^{\sharp}\left(d(\varphi^{-1})^l_{\phantom{l}m}\right)
\end{align}
as follows immediately from the chain rule. In this context, Jacobi's formula
reads as follows, with $X\in\mT M$ of definite parity.
\begin{align}
\label{eqnJacobiFormula}
X(\sdet d\varphi)=(-1)^{m+\abs{X}(m+n)}\sdet d\varphi\cdot(d\varphi^{-1})^m_{\phantom{m}n}\cdot X(d\varphi^n_{\phantom{n}m})
\end{align}
The even case was proved as Lem. 2.4 of \cite{Gro14b}, and the odd case follows
e.g. by means of introducing an additional odd generator.
(\ref{eqnJacobiFormula}) has the following corollary.

\begin{Lem}
\label{lemJacobiSum}
Let $\varphi:M\rightarrow N$ be invertible. Then
\begin{align*}
\sum_j\partial_j\left(\sdet d\varphi\cdot(d\varphi^{-1})^j_{\phantom{j}k}\right)=0
\end{align*}
\end{Lem}

\begin{proof}
We use (\ref{eqnJacobiFormula}) to calculate term $(i)$ in
\begin{align*}
\partial_j(\sdet d\varphi(d\varphi^{-1})^{j}_{\phantom{j}k})
=\partial_j(\sdet d\varphi)\cdot(d\varphi^{-1})^{j}_{\phantom{j}k}
+\sdet d\varphi\cdot\partial_j(d\varphi^{-1})^{j}_{\phantom{j}k}
=:(i)+(ii)
\end{align*}
as follows, using (\ref{eqnDifferentialHessian}).
\begin{align*}
(i)&=(-1)^{n+jn+jm}\sdet d\varphi\cdot(d\varphi^{-1})^n_{\phantom{n}m}\partial_j(d\varphi^m_{\phantom{m}n})(d\varphi^{-1})^{j}_{\phantom{j}k}\\
&=(-1)^{n+nm}\sdet d\varphi\cdot(d\varphi^{-1})^n_{\phantom{n}m}\partial_n(d\varphi^m_{\phantom{m}j})(d\varphi^{-1})^{j}_{\phantom{j}k}\\
&=-\sdet d\varphi\cdot\partial_n(d\varphi^{-1})^n_{\phantom{n}m}d\varphi^m_{\phantom{m}j}(d\varphi^{-1})^{j}_{\phantom{j}k}\\
&=-\sdet d\varphi\cdot\partial_n(d\varphi^{-1})^n_{\phantom{n}k}\\
&=-(ii)
\end{align*}
\end{proof}

Let $\varphi:\bR^{n|m}\rightarrow\bR^{n|m}$ be a change of coordinates
$\zeta^j=\varphi^{\sharp}(\xi^j)$, symbolically denoted
$\varphi:\zeta\rightarrow\xi$.
By construction, the transformation of coordinate expressions of
(co)vectors under $\varphi$ is governed by (\ref{eqnIsoPullback}).
By means of the local expression of the differential, this can be rewritten
\begin{align}
\label{eqnCoordinateVFTrafo}
\varphi^{\star}\partial_{\xi^k}
=\partial_{\zeta^m}\cdot(d\varphi^{-1})^m_{\phantom{m}k}\;,\qquad
\varphi^{\star}d\xi^j=(-1)^{j+jm}d\zeta^m\cdot d\varphi^j_{\phantom{j}m}
=d\zeta^m\cdot\left(d\varphi^{ST}\right)^m_{\phantom{m}j}
\end{align}
Here, $d:\mO_M\rightarrow\mT^*M$ is the differential on superfunctions,
which we define with the convention
\begin{align*}
df[X]:=(-1)^{\abs{f}\abs{X}}X(f)\;,\qquad\mathrm{for}\quad X\in\mT M
\end{align*}
As usual, the superscript ''ST'' in (\ref{eqnCoordinateVFTrafo}) means
supertranspose, and the indices refer to the local basis elements
$(d\xi^m)$ and $(d\zeta^k)$ of $\mT^*M$ induced by the coordinate functions.
Nota bene, such a basis is related to the standard dual basis
of the coordinate vector fields by $(\partial_{\xi^m})^*=(-1)^md\xi^m$.

\subsection{Complex Supermanifolds and Differential Forms}

In the following, we consider a complex supermanifold $M=(M_{\overline{0}},\mO_M)$
of dimension $n|m$, which we always assume connected.
By a standard construction (see \cite{HW87}), $M$ carries a canonical
structure of a real supermanifold of dimension $2n|2m$, which is
compatible with the real manifold of dimension $2n$ associated with $M_{\overline{0}}$.
In terms of (complex) local coordinates $(\xi^k)$, the picture is as follows.
With a natural notion of conjugating superfunctions, the relations
\begin{align}
\label{eqnComplexCoordinates}
\xi^k_R=\frac12(\xi^k+\overline{\xi}^k)\;,\qquad
\xi^j_I=\frac1{2i}(\xi^j-\overline{\xi}^j)
\end{align}
define real coordinates $(\xi^k_R,\xi^j_I)$.
Likewise, a morphism $\varphi:M\rightarrow N$ between complex supermanifolds
canonically induces a morphism between the corresponding real supermanifolds.
We shall use the resulting functor frequently in an implicit manner, writing
$M$ and $\varphi$ for both the complex or real supermanifold or morphism,
respectively. An analogous comment applies to the sheaves
of holomorphic and real superfunctions. By a slight abuse of
notation, we shall write $\mO_M$ in either case.

We define a holomorphic (super) vector bundle on $M$
to be a vector bundle in the general sense of the previous subsection,
on $M$ considered as a complex supermanifold.
Likewise, a smooth (super) vector bundle on $M$ is a vector bundle
on $M$ considered as a real supermanifold.
There is a forgetful functor from holomorphic vector bundles
to smooth vector bundles.

The tangent bundle of a supermanifold, complex or real, was defined
in the previous subsection as the sheaf of superderivations with
respect to the sheaf of superfunctions.
In the present context, there are two variants thereof, according
to the complex or real pictures of $M$.
We adopt the following convention. The real and complex tangent bundles
shall be denoted by $\mT M$ and $\mT^{1,0}M$, respectively.
They are related to each other, analogous to the case of classical
complex manifolds, as follows. The complex structure of $M$ induces a
complex structure $J$ on the smooth bundle $\mT M$.
The eigensheaf with eigenvalue $+1$ of the complexification $\mT M\otimes\bC$
can be identified with the complex tangent bundle $\mT^{1,0}M$.
As usual, we denote the $-1$-eigensheaf by $\mT^{0,1}M$.
There is a canonical complex-linear isomorphism
\begin{align}
\label{eqnJmath}
\jmath:(\mT M,J)\rightarrow(\mT^{1,0}M,i)\;,\qquad X\mapsto\frac12(X-iJX)
\end{align}
In terms of complex coordinates $(\xi^k)$ with induced real coordinates
(\ref{eqnComplexCoordinates}), the respective canonical derivations
are related by
\begin{align}
\label{eqnLocalJmath}
\jmath\left(\partial_{\xi^k_R}\right)=\partial_{\xi^k}\;,\qquad
\jmath\left(\partial_{\xi^k_I}\right)=i\partial_{\xi^k}
\end{align}

In the rest of this subsection, we state some properties
of bundles of differential forms. Up to signs, this part largely parallels
the ungraded case and is, therefore, kept short.
The prescription
\begin{align*}
\Omega^{p,q}M:=\bigwedge^p(\mT^{1,0}M)^*\otimes\bigwedge^q(\mT^{0,1}M)^*
\end{align*}
defines smooth vector bundles over $M$
which are referred to as bundles of differential forms.
On $\Omega^{p,q}M$, there is a natural notion of exterior derivative,
and its canonical projection to $\Omega^{p,q+1}M$ is denoted $\dbar$.
We are mostly interested in the case $p=0$.
A section can then locally be written as a sum of elements of the form
$d\oxi^I\cdot f$, where $(\xi^k)$ are complex coordinates of $M$
with conjugation as in (\ref{eqnComplexCoordinates}),
and $I$ denotes some multiindex. The $\dbar$-operator then reads
\begin{align*}
\dbar\left(d\oxi^I\cdot f\right)
=\sum_k\left((-1)^{\abs{\xi^k}\abs{\xi^I}}d\oxi^k\wedge d\oxi^I\cdot\dd[f]{\xi^k}\right)
\end{align*}

We also consider differential forms with values in some holomorphic
vector bundle $\mE$, that is bundles of the form
\begin{align*}
\Omega^{0,q}(\mE):=\Omega^{0,q}M\otimes\mE
\end{align*}
Define the $\dbar$-operator on such bundles via
\begin{align}
\label{eqnDbar}
\dbar(\alpha\otimes e):=\dbar(\alpha)\otimes e
\end{align}
This is well-defined, since $e$ transforms holomorphically (annihilated by $\dbar$).

Of particular interest will be multivector fields $\mE=\bigwedge^p\mT^{1,0}M$.
Throughout, we adopt the degree conventions explained in Sec.
\ref{secBV}. In particular, the cohomological degree of
a section $\alpha\in\Omega^{0,q}(\bigwedge^p\mT^{1,0}M)$ is $\deg(\alpha)=p+q$.
The Schouten-Nijenhuis bracket, to be defined next, is a natural
superisation of the classical bracket bearing that name (see e.g.
Chp. 6 of \cite{Huy05} for a formula)
It extends the vector field bracket.

\begin{Def}[Schouten-Nijenhuis bracket]
\label{defSchoutenBracket}
Let $\scal[[]{\cdot}{\cdot}$ be the map
\begin{align*}
\scal[[]{\cdot}{\cdot}:\Omega^{0,q}\left(\bigwedge^p\mT^{1,0}M\right)\times\Omega^{0,q'}\left(\bigwedge^{p'}\mT^{1,0}M\right)\longrightarrow\Omega^{0,q+q'}\left(\bigwedge^{p+p'-1}\mT^{1,0}M\right)
\end{align*}
defined as follows.
For $f\in\mO_M$ and $v_i,w_j\in\mT^{1,0}X$, we set
\begin{align*}
\scal[[]{f}{v_1\wedge\ldots\wedge v_p}:&=-(-1)^{(i-1)+\abs{v_i}\left(\sum_{l=1}^{i-1}\abs{v_l}+\abs{f}\right)}v_i(f)\cdot v_1\wedge\ldots\wedge\widehat{v_i}\wedge\ldots\wedge v_p\\
\scal[[]{v_1\wedge\ldots\wedge v_p}{f}:&=-(-1)^{(p+1)+\abs{f}\left(\abs{v_1}+\ldots+\abs{v_p}\right)}\scal[[]{f}{v_1\wedge\ldots\wedge v_p}
\end{align*}
and
\begin{align*}
&\scal[[]{w_1\wedge\ldots\wedge w_{p'}}{v_1\wedge\ldots\wedge v_p}\\
&\qquad:=\sum_{j,i}(-1)^{M_{(i,j)}}\scal[[]{w_j}{v_i}\wedge(w_1\wedge\ldots\wedge\widehat{w_j}\wedge\ldots)\wedge(v_1\wedge\ldots\widehat{v_i}\wedge\ldots)
\end{align*}
with the sign
\begin{align*}
M_{(i,j)}:=i+j+\abs{w_j}\sum_{l=1}^{j-1}\abs{w_l}+\abs{v_i}\left(\sum_{l=1}^{i-1}\abs{v_l}+\sum_{l=1}^{p'}\abs{w_l}+\abs{w_j}\right)
\end{align*}
With $(\xi^k)$ local coordinates and $I$,$J$,$K$,$L$ denoting multiindices, we
further prescribe
\begin{align*}
&\scal[[]{fd\oxi^I\otimes\partial_{\xi^J}}{gd\oxi^K\otimes\partial_{\xi^L}}\\
&\qquad\qquad\qquad:=(-1)^{q'(p+1)+\abs{f}\abs{\xi^I}+\abs{\xi^K}(\abs{g}+\abs{\xi^J}+\abs{f})}
d\oxi^I\wedge d\oxi^K\otimes\scal[[]{f\partial_{\xi^J}}{g\partial_{\xi^L}}
\end{align*}
where $q=\abs{I}$, $p=\abs{J}$, $q'=\abs{K}$, $p'=\abs{L}$.
\end{Def}

Although written in terms of
coordinates $(\xi^k)$ on $M$, the resulting bracket is welldefined
since the transformation of $d\oxi^I$ gets annihilated by $v\in\mT^{1,0}M$.

\begin{Lem}
\label{lemSchoutenSymmetry}
The Schouten-Nijenhuis bracket satisfies the following symmetry.
\begin{align*}
\scal[[]{\alpha}{\beta}=-(-1)^{(\deg(\alpha)+1)(\deg(\beta)+1)+\abs{\alpha}\abs{\beta}}\scal[[]{\beta}{\alpha}
\end{align*}
\end{Lem}

\begin{proof}
This follows from the definition by a straightforward calculation.
\end{proof}

\begin{Def}
\label{defExteriorProduct}
Exterior product:
\begin{align*}
\Omega^{0,p}\left(\bigwedge^r\mT^{1,0}M\right)\times\Omega^{0,q}\left(\bigwedge^s\mT^{1,0}M\right)
\rightarrow\Omega^{0,p+q}\left(\bigwedge^{r+s}\mT^{1,0}M\right)
\end{align*}
defined as follows, for $\oalpha\in\Omega^{0,p}M$, $\beta\in\Gamma(\bigwedge^r\mT^{1,0}M)$, $\ogamma\in\Omega^{0,q}M$ and $\delta\in\bigwedge^s\mT^{1,0}M$.
\begin{align*}
(\oalpha\otimes\beta)\wedge(\ogamma\otimes\delta)
:=(-1)^{rq+\abs{\beta}\abs{\ogamma}}(\oalpha\wedge\ogamma)\otimes(\beta\wedge\delta)
\end{align*}
\end{Def}

\begin{Lem}
\label{lemSchoutenDerivation}
Let $\alpha\in\Omega^{0,q}(\bigwedge^p\mT^{1,0}M)$. The map
\begin{align*}
\scal[[]{\alpha}{\cdot}:\Omega^{0,q'}\left(\bigwedge^{p'}\mT^{1,0}M\right)\longrightarrow\Omega^{0,q+q'}\left(\bigwedge^{p+p'-1}\mT^{1,0}M\right)
\end{align*}
is a derivation. More precisely,
\begin{align*}
\scal[[]{\alpha}{\beta\wedge\gamma}
=\scal[[]{\alpha}{\beta}\wedge\gamma
+(-1)^{(\deg(\alpha)+1)\deg(\beta)+\abs{\alpha}\abs{\beta}}\beta\wedge\scal[[]{\alpha}{\gamma}
\end{align*}
\end{Lem}

\begin{proof}
This follows from a lengthy, but otherwise straightforward, calculation
in local expressions, using the definitions. One proves the case
$\alpha,\beta,\gamma\in\Omega^{0,0}(\bigwedge\mT^{1,0}M)$ first
and deduces from this the general case.
\end{proof}

\subsection{Connections, Parallel Transport and Holonomy}

Let $M$ continue to denote a complex supermanifold with associated
real supermanifold referred to by $M$, too.
A connection on a smooth vector bundle $\mE$ on $M$ is an even real-linear
sheaf morphism
\begin{align*}
\nabla:\mE\rightarrow\mT M^*\otimes_{\mO_M}\mE\;,\qquad
\nabla(fe)=df\otimes_{\mO_M}e+f\cdot\nabla e\qquad\mathrm{for}\quad f\in\mO_M
\end{align*}
with $\mO_M$ denoting the real superfunctions.
In the case of a holomorphic vector bundle, we refer to a connection
as a connection in the above sense on the associated smooth vector bundle.
In particular, in our point of view, we do not demand complex linearity
but, if present, consider it as an extra structure.

A connection $\nabla$ on the tangent bundle $\mE=\mT M$ induces a connection
on $\mT^{1,0}M$,
in the following referred to by the same symbol, through the prescription
\begin{align}
\label{eqnJmathNabla}
X\mapsto\jmath\circ\nabla_{\cdot}\circ\jmath^{-1}(X)
\end{align}
involving the canonical isomorphism (\ref{eqnJmath}).
This connection is complex linear precisely if $\nabla J=0$ (K\"ahler case).
In this case, it coincides with the complex-linear
extension $\nabla^{\bC}$ restricted to $\mT^{1,0}M\subseteq\mT^{\bC}M$.

We shall need the transformation behaviour of a connection $\nabla$ on $\mT M$
under a transformation of coordinates $\varphi:\zeta\rightarrow\xi$.
In terms of coordinates $(\xi^i)$, the connection is determined by its
left or right Christoffel symbols
\begin{align}
\label{eqnChristoffelSymbols}
\nabla_{\partial_{\xi^k}}\partial_{\xi^l}=\Gamma^q_{kl}\cdot\partial_{\xi^q}=\partial_{\xi^q}\cdot ^R\Gamma^q_{kl}
\end{align}
The pullback under a (local) isomorphism $\varphi$, as demanded by
(\ref{eqnIsoPullback}), is defined as follows
\begin{align}
\label{eqnConnectionPullback}
(\varphi^{\star}\nabla)_{\varphi^{\star}X}(\varphi^{\star}Y):=\varphi^{\star}(\nabla_XY)
\end{align}
By the discussion leading to (\ref{eqnCoordinateVFTrafo}), equation
(\ref{eqnConnectionPullback}) describes the transformation of
the local expression for $\nabla$
in case $\varphi:\zeta\rightarrow\xi$ is a change of coordinates.
The Christoffel symbols with respect to $\zeta$ are then given by
\begin{align*}
\hat{\Gamma}^a_{bc}\cdot\partial_{\zeta^a}=(\varphi^{\star}\nabla)_{\partial_{\zeta^a}}\partial_{\zeta^b}
\end{align*}
They are related to $\Gamma$ by the following result.

\begin{Lem}
\label{lemChristoffelTrafo}
Let $\nabla$ be a connection on $\mT M$ and $\varphi:\zeta\rightarrow\xi$
be a coordinate transformation. Then, with the notation introduced above,
\begin{align*}
&\varphi^{\sharp}(\Gamma_{kl}^p)=(-1)^{m(m+k)+ql+p(p+q)}\\
&\qquad\qquad\qquad(d\varphi^{-1})^m_{\phantom{m}k}\cdot\left((-1)^{qn}\hat{\Gamma}^q_{mn}(d\varphi^{-1})^n_{\phantom{n}l}+(-1)^q\partial_{\xi^m}(d\varphi^{-1})^q_{\phantom{q}l}\right)\cdot(d\varphi)^p_{\phantom{p}q}
\end{align*}
\end{Lem}

\begin{proof}
This follows from a straightforward calculation involving
(\ref{eqnCoordinateVFTrafo}) and Leibniz' rule.
\end{proof}

Returning to the general case of a connection $\nabla$ on a vector bundle $\mE$,
we mention the notion of parallel transport
$P_{\gamma}:x^*\mE\rightarrow y^*\mE$ along an $S$-path
$\gamma:S\times[0,1]\rightarrow M$ connecting $S$-points $x,y:S\rightarrow M$.
Here, $S$ is a parameter space which we restrict to the case of superpoints
$S=\bR^{0|L}$. This is detailed in \cite{Gro14c}.
In terms of local coordinates $(\xi^k)$ and a trivialisation $(T^m)$ of $\mE$,
the parallelness condition $(\gamma^*\nabla)_{\partial_t}P_{\gamma}[X_x]=0$
for $X_x\in x^*\mE$ reads as follows.
\begin{align}
\label{eqnParallelTransport}
\partial_t(P_{\gamma})^m_{\phantom{m}p}
=-(-1)^{m(k+1)}\partial_t(\gamma^*(\xi^l))\cdot\gamma^*(\Gamma^m_{lk})\cdot(P_{\gamma})^k_{\phantom{k}p}
\end{align}
In this equation, the Christoffel symbols are defined analogous to
(\ref{eqnChristoffelSymbols}).

\begin{Lem}
Let $P_{\gamma}$ denote parallel transport of a connection on $\mT M$
along some $S$-path $\gamma$. Then parallel transport with respect to
the induced connection (\ref{eqnJmathNabla}) on $\mT^{1,0}$ is given by
$\jmath\circ P_{\gamma}\circ\jmath^{-1}$.
\end{Lem}

A meaningful notion of holonomy in supergeometry was first introduced in
\cite{Gal09} via a Harish-Chandra superpair construction,
while a categorical approach was developed in \cite{Gro14c}.
As analysed in \cite{Gro16}, both theories are equivalent,
although in a nontrivial fashion through the Twofold Theorem
and the Comparision Theorem established in that reference.
In this article, we utilise the categorical approach that we shall sketch
in the following.

For a fixed $S$-point $x$, the holonomy group $\Hol_x$ is defined as the
set of parallel transport operators $P_{\gamma}$ with $\gamma$ a piecewise
smooth $S$-loop starting and ending in $x$.
It can be shown to carry the structure of a Lie group.
By a theorem of Ambrose-Singer type, its Lie algebra $\hol_x$ is
generated by endomorphisms of the form
\begin{align}
\label{eqnHolonomyAlgebra}
\{P_{\gamma}^{-1}\circ\scal[R_y]{u}{v}\circ P_{\gamma}\setsep
y:S\rightarrow M\,,\;\gamma:x\rightarrow y\;\;\mathrm{pw. smooth}\,,\;
u,v\in(y^*\mT M)_{\overline{0}}\}
\end{align}
where $R$ denotes the curvature tensor with respect to $\nabla$.
In particular, $\Hol_x$ is trivial for a flat connection (with $R=0$)
on a simply-connected supermanifold.

As it stands, $\Hol_x$ alone does not contain enough information
for a good holonomy theory.
To that end, let $T=\bR^{0|L'}$ be another superpoint and consider $x$
as an $S\times T$-point, denoted $x_T:S\times T\rightarrow M$.
The prescription $T\mapsto\Hol_x(T):=\Hol_{x_T}$ extends to
a Lie group valued functor, referred to as the holonomy group functor $\Hol_x$
The reader should note that, in general, this functor is not representable.
In particular, it cannot be identified with the $\Lambda$-point functor
of Schwarz \cite{Shv84} and Voronov \cite{Vor84}, by means of which
a supermanifold is characterised.

For the purposes of the present article it suffices to consider the holonomy
group functor with respect to a topological point $x:\bR^{0|0}\rightarrow M$.
The Holonomy Principle can then be cast in the following form.
Recall that all supermanifolds occurring are assumed connected.

\begin{Thm}[Holonomy Principle]
\label{thmHolonomyPrinciple}
The pullback $X_x:=x^*X\in x^*\mE$ of a parallel global section
$X\in\mE$ with $\nabla X\equiv 0$ is holonomy invariant $\Hol_x(T)\cdot X_x=X_x$
for every $T=\bR^{0|L'}$.
Conversely, invariance in this sense of a section $X_x\in x^*\mE$
implies the unique existence of a parallel global section $X\in\mE$ such that
$x^*X=X_x$.
\end{Thm}

\subsection{The Canonical Bundle and Integral Forms}
\label{secCanonicalBundle}

The Berezinian of a free $A$-supermodule $M$ of rank $p|q$ for $A$ a
supercommutative superalgebra is the free $A$-supermodule of rank
$1|0$ (for $q$ even) or $0|1$ (for $q$ odd) defined through a distinguished
class of bases $[s]=[s_1\ldots s_{p+q}]$ for $(s_i)$ an $M$-basis
and the relation $[s']=[s]\cdot\Ber\varphi$ if $(s'_i)=(s_i)\cdot\varphi$,
see e.g. Chp. 3 of \cite{Man88} for details.
This construction carries over to vector bundles through local trivialisations.
The essential notion for the present paper is the following.

\begin{Def}
Let $M$ be a complex supermanifold. Its \emph{canonical bundle}
is the Berezinian of the complex cotangent sheaf
$\Ber M:=\Ber(\mT^{1,0}M)^*$.
\end{Def}

As mentioned above, any choice of coordinates $(\xi^k)$ on $M$
induces a local covector basis
$(d\xi^k)$ of $(\mT^{1,0}M)^*$.
We shall use the following notation for the induced Berezinian section.
\begin{align*}
[d\xi]:=[d\xi^1\ldots d\xi^{p+q}]
\end{align*}
By (\ref{eqnCoordinateVFTrafo}), this transforms under a coordinate transformation
$\varphi:\zeta\rightarrow\xi$ as follows.
\begin{align*}
\varphi^{\star}[d\xi]:=[\varphi^{\star}d\xi]=[d\zeta\cdot d\varphi^{ST}]=\sdet d\varphi^{ST}\cdot[d\zeta]=\sdet d\varphi\cdot[d\zeta]
\end{align*}

\begin{Def}
\label{defIntegralForms}
The sheaves of integral forms are defined by
\begin{align*}
I^{n-p}:=\bigwedge^p\mT^{1,0}M\otimes\Ber M
\end{align*}
for $p\geq 0$.
\end{Def}

There is a real counterpart of the previous definition, which plays a role
for the theory of integration similar to that of the de Rham-complex
in classical geometry, see e.g. Chp. 3 of \cite{DM99}.
For the present complex case, there is no immediate reference to integration.
Nevertheless, we refer to the $I^{n-p}$ as 'integral' forms
for the otherwise analogous structure.
There is a natural operator as follows.

\begin{Def}
\label{defPartialIntegralForms}
We define $\partial:I^{n-p}\rightarrow I^{n-p+1}$ by
\begin{align*}
&\partial\left(f\cdot\dd{\xi^1}\wedge\ldots\wedge\dd{\xi^p}\otimes[d\xi]\right)\\
&\qquad\qquad:=\sum_{i=1}^{n+m}(-1)^{M_i}\dd[f]{\xi^i}\cdot\left(\dd{\xi^1}\wedge\ldots\wedge\dd{\xi^{i-1}}\wedge\widehat{\dd{\xi^i}}\wedge\dd{\xi^{i+1}}\wedge\ldots\wedge\dd{\xi^p}\right)\otimes[d\xi]
\end{align*}
with respect to (complex) local coordinates $\xi^i$ on $M$.
As is usual, the hat symbol means omission. Moreover, the sign $M_i$ reads
\begin{align*}
M_i=(i-1)+\abs{\xi^i}\cdot\left(\sum_{j=1}^{i-1}\abs{\xi^j}+\abs{f}\right)
\end{align*}
as arising from moving $\dd{\xi^i}$ to the front.
\end{Def}

\begin{Lem}
\label{lemPartialWelldefined}
$\partial$ is well-defined (independent of coordinates), and $\partial^2=0$.
\end{Lem}

The operator $\partial$ should be thought of as a suitable generalisation of
the classical operator $\partial:\Omega^{p,q}\rightarrow\Omega^{p+1,q}$
(with $q=0$) on a complex manifold. Indeed, in this case there is a canonical
isomorphism $I^{n-p}\cong\Omega^{p,0}$, and the local expression for $\delta$
is as stated. In the case at hand of a complex supermanifold, the bundles
$I^{n-p}\neq\Omega^{q,0}$ are different, and a welldefined operator
of the local form stated can only be defined on the former, in the way
presented here. Lem. \ref{lemPartialWelldefined} can be proved by a lengthy
but direct calculation
involving the explicit expressions under a coordinate transformation
of the objects occurring, and using Lem. \ref{lemJacobiSum}.

Let $\dbar$ denote the operator on $\Omega^{0,q}(I^{n-l})$
as defined by (\ref{eqnDbar}).
Moreover, we define
\begin{align}
\label{eqnPartial}
\partial:\Omega^{0,q}(I^{n-p})\rightarrow\Omega^{0,q}(I^{n-p+1})\;,\qquad
\partial(d\oxi^I\otimes e):=(-1)^{q}d\oxi^I\otimes\partial e
\end{align}
where $I$ is a multiindex with $\abs{I}=q$ and $e\in I^{n-p}$.
The operator $\partial$ on the right hand side is the map from
Def. \ref{defPartialIntegralForms}. Although written in terms of
coordinates $(\xi^k)$ on $M$, the resulting operator is welldefined
since the transformation of $d\oxi^I$ gets annihilated by $\partial$.
It follows immediately that
\begin{align}
\label{eqnPartialSquared}
\partial^2=0
\end{align}

\begin{Lem}
\label{lemPartialDbar}
$\partial$ anticommutes with $\dbar$:
\begin{align*}
\partial\dbar=-\dbar\partial
\end{align*}
\end{Lem}

We remark that the sign $(-1)^q$ in (\ref{eqnPartial}) is important
for this lemma to hold.

\begin{proof}
This follows from a straightforward calculation involving the
expressions of either operator in local coordinates.
\end{proof}

We close this appendix with a comparison to Manin's theory of integral
forms as presented in Sec. 4.5 of \cite{Man88}. It turns out that the
operator $\delta$ defined in paragraph 4 of that reference coincides with
our operator $\partial$, upon considering suitable identifications
as to be detailed in the following. In particular,
Lem. \ref{lemPartialWelldefined} then follows from analogous properties
of $\delta$.

To begin with let, in general,
$A$ be a supercommutative superalgebra and $V$ be an
$A$-supermodule. Let $\Pi V$ and $V\Pi$ denote the supermodules with reversed
parity and, respectively, same right and left $A$-multiplication as $V$.
Let $TV$ denote the tensor algebra of $V$.
We define a map $\Gamma:TV\rightarrow T(V\Pi)$ through the prescription
\begin{align*}
\Gamma(a):=a\;,\qquad
\Gamma(v_1\otimes\ldots\otimes v_p)
:=(-1)^{\sum_{i=1}^p(p-i+1)\abs{v_i}}(v_1\Pi)\otimes\ldots\otimes(v_p\Pi)
\end{align*}
for $a\in A$ and $v_i\in V$. This is well-defined, $A$-linear and bijective.
Restricting to the exterior algebra yields a well-defined map, still denoted
$\Gamma$, as follows.
\begin{align*}
\Gamma:\bigwedge(V)\rightarrow S(V\Pi)\;,\qquad
\abs{\Gamma|_{\bigwedge^pV}}=\sum_{i=1}^p\abs{v_i}+\sum_{i=1}^p(\abs{v_i}+1)=p
\end{align*}
This map is an isomorphism (in the weaker sense of non-parity preserving)
of $A$-supermodules, but not of superalgebras. We may form the tensor product
of either side with another $A$-module $W$ and consider the trivial extension
of $\Gamma$, that we shall denote by the same symbol.

In our case of interest, this yields the bundle map
\begin{align*}
\Gamma:\Ber M\otimes\bigwedge\mT^{1,0}M\rightarrow\Ber M\otimes S\left(\mT^{1,0}M\Pi\right)
\end{align*}
defined through
\begin{align*}
&\Gamma\left([d\xi]\cdot f\otimes\dd{\xi^1}\wedge\ldots\wedge\dd{\xi^p}\right)\\
&\qquad\qquad\qquad=(\id\otimes\Gamma)\left([d\xi]\cdot f\otimes\dd{\xi^1}\wedge\ldots\wedge\dd{\xi^p}\right)\\
&\qquad\qquad\qquad=(-1)^{p(m+\abs{f})+\sum_{i=1}^p(p-i+1)\abs{\xi^i}}[d\xi]\cdot f\otimes\dd{\xi^1}\Pi\odot\ldots\odot\dd{\xi^p}\Pi
\end{align*}
where $m=\abs{[d\xi]}$ is the odd dimension $\dim_{\bC}M=n|m$,
and the sign $p(m+\abs{f})$ comes from commuting $\Gamma$ past $[d\xi]\cdot f$.
On the right hand side, there are the integral forms (in the complex case)
in the sense of \cite{Man88}.
On the left, we have the integral forms as in Def. \ref{defIntegralForms},
with the general canonical isomorphism $V\otimes W\cong W\otimes V$ for
$A$-supermodules $V$ and $W$, defined through $v\otimes w\mapsto (-1)^{\abs{v}\abs{w}}w\otimes v$, left implicit.

\begin{Lem}
\label{lemManinOperator}
Let $\delta:\Ber M\otimes S(\mT^{1,0}M\Pi)\circlearrowleft$ denote the operator
defined in Sec.4.5 of \cite{Man88}. Then
$\delta\circ\Gamma=-\Gamma\circ\partial$.
\end{Lem}

\begin{proof}
Both operators are defined in terms of local coordinates.
By a direct calculation, they coincide.
\end{proof}

As detailed in \cite{Man88}, the operator $\delta$ is welldefined and
satisfies $\delta^2=0$. By the preceding lemma, it is clear that
$\partial$ has analogous properties. This gives a proof of Lem. \ref{lemPartialWelldefined}.

\bibliographystyle{alpha}

\addcontentsline{toc}{section}{References}

\begin{thebibliography}{Gro14b}

\bibitem[BK98]{BK98}
S.~Barannikov and M.~Kontsevich.
\newblock {F}robenius manifolds and formality of {L}ie algebras of polyvector
  fields.
\newblock {\em Int. Math. Res. Notices}, 1998(4):201--215, 1998.

\bibitem[BV81]{BV81}
I.~Batalin and G.~Vilkovisky.
\newblock Gauge algebra and quantization.
\newblock {\em Phys. Lett. B}, 102(1):27--31, 1981.

\bibitem[CCF11]{CCF11}
C.~Carmeli, L.~Caston, and R.~Fioresi.
\newblock {\em Mathematical Foundations of Supersymmetry}.
\newblock European Mathematical Society, 2011.

\bibitem[DM99]{DM99}
P.~Deligne and J.~Morgan.
\newblock Notes on supersymmetry.
\newblock In P.~Deligne et~al., editor, {\em Quantum Fields and Strings: A
  Course for Mathematicians}. American Mathematical Society, 1999.

\bibitem[Gal09]{Gal09}
A.~Galaev.
\newblock Holonomy of supermanifolds.
\newblock {\em Abhandlungen aus dem Mathematischen Seminar der Universität
  Hamburg}, 79:47--78, 2009.

\bibitem[Gro14a]{Gro14b}
J.~Groeger.
\newblock Divergence theorems and the supersphere.
\newblock {\em J. Geom. Phys.}, 77, 2014.

\bibitem[Gro14b]{Gro14c}
J.~Groeger.
\newblock Super {W}ilson loops and holonomy on supermanifolds.
\newblock {\em Comm. Math.}, 22(2), 2014.

\bibitem[Gro16]{Gro16}
J.~Groeger.
\newblock The twofold way of super holonomy.
\newblock {\em Forum Math.}, Ahead of print, 2016.

\bibitem[Han12]{Han12}
F.~Hanisch.
\newblock Variational problems on supermanifolds.
\newblock Dissertation, Universität Potsdam, 2012.

\bibitem[Huy05]{Huy05}
D.~Huybrechts.
\newblock {\em Complex Geometry}.
\newblock Springer, 2005.

\bibitem[HW87]{HW87}
C.~Haske and R.~Wells.
\newblock Serre duality on complex supermanifolds.
\newblock {\em Duke Mathematical Journal}, 54(2), 1987.

\bibitem[Joy00]{Joy00}
D.~Joyce.
\newblock {\em Compact manifolds with special holonomy}.
\newblock Oxford University Press, 2000.

\bibitem[Lei80]{Lei80}
D.~Leites.
\newblock Introduction to the theory of supermanifolds.
\newblock {\em Russian Math. Surveys}, 35(1), 1980.

\bibitem[Man88]{Man88}
Y.~Manin.
\newblock {\em Gauge Field Theory and Complex Geometry}.
\newblock Springer, 1988.

\bibitem[Man99]{Man99}
Y.~Manin.
\newblock {\em Frobenius Manifolds, Quantum Cohomology, and Moduli Spaces}.
\newblock AMS, 1999.

\bibitem[RW05]{RW05}
M.~Ro\v{c}ek and N.~Wadhwa.
\newblock On {C}alabi-{Y}au supermanifolds.
\newblock {\em Adv. Theor. Math. Phys.}, 9(2):315--320, 2005.

\bibitem[Sch98]{Sch98}
V.~Schechtman.
\newblock Remarks on formal deformations and {B}atalin-{V}ilkovisky algebras.
\newblock Preprint, MPI Bonn, 1998.

\bibitem[Shv84]{Shv84}
A.~Shvarts.
\newblock On the definition of superspace.
\newblock {\em Teoret. Mat. Fiz.}, 60(1):37--42, 1984.

\bibitem[Ten75]{Ten75}
B.~Tennison.
\newblock {\em Sheaf Theory}.
\newblock Cambridge University Press, 1975.

\bibitem[Tia87]{Tia87}
G.~Tian.
\newblock Smoothness of the universal deformation space of compact
  {C}alabi-{Y}au manifolds and its {P}etersson-{W}eil metric.
\newblock In S.~Yau, editor, {\em Mathematical Aspects of String Theory},
  volume~1, pages 629--646. World Scientific, 1987.

\bibitem[Tod89]{Tod89}
A.~Todorov.
\newblock The {W}eil-{P}etersson geometry of the moduli space of {$SU(n\geq
  3)$} ({C}alabi-{Y}au) manifolds {I}.
\newblock {\em Comm. Math. Phys.}, 126:325--346, 1989.

\bibitem[Var04]{Var04}
V.~Varadarajan.
\newblock {\em Supersymmetry for Mathematicians: An Introduction}.
\newblock American Mathematical Society, 2004.

\bibitem[Vor84]{Vor84}
A.~Voronov.
\newblock Mappings of supermanifolds.
\newblock {\em Teoret. Mat. Fiz.}, 60(1):43--48, 1984.

\end{thebibliography}

\end{document}